\newif\ifreport\reporttrue
\ifreport
\documentclass[journal]{IEEEtran}
\else
\documentclass[10pt, conference, letterpaper]{IEEEtran}
\fi
\makeatletter
\def\ps@headings{%
\def\@oddhead{\mbox{}\scriptsize\rightmark \hfil \thepage}%
\def\@evenhead{\scriptsize\thepage \hfil \leftmark\mbox{}}%
\def\@oddfoot{}%
\def\@evenfoot{}}
\makeatother
\usepackage{amsfonts}
\usepackage{amsthm}
\usepackage{amsmath}
\usepackage{amssymb}
\usepackage{bm}
\usepackage{cite}
\usepackage{graphicx}
\usepackage[tight,footnotesize]{subfigure}
\usepackage{graphics,booktabs,color,epsfig,subfigure}
\usepackage[singlelinecheck=false]{caption}

\theoremstyle{plain}
\newtheorem{example}{Example}
\newtheorem{lemma}{Lemma}

\newtheorem{Theorem}{Theorem}
\theoremstyle{definition}

\newtheorem{definition}{Definition}

\theoremstyle{remark}

\usepackage{color}

\newcommand{\ignore}[1]{}

\begin{document}
\IEEEoverridecommandlockouts
\title{Provably Delay Efficient Data Retrieving in Storage Clouds}

\author{Yin Sun$^\dag$, Zizhan Zheng$^\dag$, C. Emre Koksal$^\dag$, Kyu-Han Kim$^\ddag$, and Ness B. Shroff$^\dag{}^\S$ \\
$^\dag$Dept. of ECE, $^\S$Dept. of CSE, The Ohio State University, Columbus, OH\\
$^\ddag$Hewlett Packard Laboratories, Palo Alto, CA
\thanks{This work has been supported in part by an IRP grant from HP.}
}


\maketitle

%

\begin{abstract}

One key requirement for storage clouds is to be able to retrieve data quickly. Recent system measurements have shown that the data retrieving delay in storage clouds is highly variable, which may result in a long latency tail. One crucial idea to improve the delay performance is to retrieve multiple data copies by using parallel downloading threads. However, how to optimally schedule these downloading threads to minimize the data retrieving delay remains to be an important open problem. In this paper, we
develop low-complexity thread scheduling policies for several important classes of data downloading time distributions, and prove that these policies are either delay-optimal or within a constant gap from the optimum delay performance.
These theoretical results hold for an arbitrary arrival process of read requests that may contain finite or infinite read requests, and for heterogeneous MDS storage codes that can support diverse storage redundancy and reliability requirements for different data files.
 Our numerical results show that 
the delay performance of the proposed policies is significantly better than that of First-Come-First-Served (FCFS) policies considered in prior work.

\end{abstract}

%
%

\section{Introduction}
Cloud storage is a prevalent solution for online data storage, as it provides the appealing benefits of easy access, low maintenance, elasticity, and scalability.
The global cloud storage market is expected to reach \$56.57 billion by 2019, with a compound annual growth rate of 33.1\% \cite{CloudStorageReport}.

{In cloud storage systems, multiple copies of data are generated using simple replications \cite{Google2003,Chang:2006:BDS:1298455.1298475,Hadoop2010} or erasure storage codes \cite{Dimakis2011,Rashmi2013,HDFS-RAID,Swift}, and distributedly stored in disks, in-memory databases and caches. For an $(n, k)$ erasure code $(n > k)$, data is divided into $k$ equal-size chunks, which are then encoded into $n$ chunks and stored in $n$ distinct storage devices. If the code satisfies the typical maximum distance separable (MDS) property, any $k$ out of the $n$ chunks are sufficient to restore original data. When $k=1$, the $(n, k)$ erasure code reduces to the case of data replication (aka repetition codes).

Current storage clouds jointly utilize multiple erasure codes to support diverse storage redundancy and reliability requirements. For instance, in Facebook's data warehouse cluster, frequently accessed data (or so called ``hot data'') is stored with 3 replicas, while rarely accessed data (``cold data'') is stored by using a more compressed (14,10) Reed-Solomon code to save space \cite{Rashmi2013}. Open-source cloud storage softwares, such as HDFS-RAID \cite{HDFS-RAID} and OpenStack Swift \cite{Swift}, have been developed to support the coexistence of multiple erasure codes.}

{One key design principle of cloud storage systems is fast data retrieval. Amazon, Microsoft, and Google all report that a slight increase in user-perceived delay will result in a concrete revenue loss \cite{Linden2006,DelayImpact2009}. However, in current storage clouds, data retrieving time is highly random and may have a long latency tail due to many reasons, including network congestion, load dynamics, cache misses, database blocking, disk I/O interference, update/maintenance activities,  and unpredictable failures \cite{Dean-CACM13,Garfinkel07anevaluation,Wang:2010,Google2003,ErasureCodingWindows}.
One important approach to curb this randomness is {\it downloading multiple data copies in parallel}. For example, if a file is stored with an $(n, k)$ erasure code, the system can schedule more than $k$ downloading ``threads'', each representing a TCP connection, to retrieve the file. The first $k$ successfully downloaded chunks are sufficient to restore the file, and the excess downloading threads are terminated to release the networking resources. By this, the retrieval latency of the file is reduced. However, scheduling redundant threads will increase the system load, which may in turn increase the latency. Such a policy provides a tradeoff between faster retrieval of each file and the extra system load for downloading redundant chunks. Therefore, a critical question is ``how to optimally manage the downloading threads to minimize average data retrieving delay?'' Standard tools in scheduling and queueing theories, e.g., \cite{Schrage68,Smith78,Michael2012book,Fox:2011:OSI:2133036.2133046,doi:10.1137/090772228} and the references therein, cannot be directly applied to resolve this challenge because they
do not allow scheduling redundant and parallel resources for service acceleration.}

In this paper, we rigorously analyze the fundamental delay limits of storage clouds. We develop low-complexity online thread scheduling policies for several important classes of data downloading time distributions, and prove that these policies are either delay-optimal or within a constant gap from the optimum delay performance.\footnote{By constant delay gap, we mean that the delay gap is bounded by a constant value that is independent of the request arrival process and system traffic load.} Our theoretical results hold for an arbitrary arrival process of read requests that may contain finite or infinite read requests, and for heterogeneous MDS storage codes that can support diverse code parameters $(n_i,k_i)$ for different data files.
The main contributions of our paper are listed as follows and summarized in Table \ref{tab2}. An interesting \emph{state evolution} argument is developed in this work, which is essential for establishing the constant delay gaps;
\ifreport
the interested reader is referred to the appendices for the detailed proofs.
\else
the interested reader is referred to the technical report \cite{tech_report2014} for the details.
\fi
\begin{table*}
\centering
\begin{tabular}{c|l|l|l|l|l||l} \hline
           & Arrival  &  Parameters of                   & Service & Downloading time &      &    \\
Theorem  & process  & MDS codes      & preemption &  distribution &  Policy & Delay gap from optimum \\
\hline
\ref{thm1} & any & $d_{\min}\geq L$ & allowed        &\emph{i.i.d.} exponential     & SERPT-R            & delay-optimal \\
\ref{thm3} & any & any   & allowed        &\emph{i.i.d.} exponential     & SERPT-R            & $\frac{1}{\mu}\sum_{l=d_{\min}}^{L-1}\frac{1}{l}$  \\
\ref{thm2} & any & $d_{\min}\geq L$ & not allowed    &\emph{i.i.d.} exponential     & SEDPT-R            & $1/\mu$ \\
\ref{thm4} & any & any   & not allowed    &\emph{i.i.d.} exponential     & SEDPT-R            & $\frac{1}{\mu}\left(\sum_{l=d_{\min}}^{L-1}\frac{1}{l}+1\right)$  \\
\ref{thm7} & any & any   & not allowed &\emph{i.i.d.} New-Longer-than-Used& SEDPT-NR  & $O(\ln L/\mu)$  \\
\ref{thm8} & any & any   & allowed &\emph{i.i.d.} New-Longer-than-Used&  SEDPT-WCR  & $O(\ln L/\mu)$  \\
\ref{thm5} & any & $k_i =1$, $d_{\min}\geq L$ & not allowed   &\emph{i.i.d.} New-Shorter-than-Used& SEDPT-R            & delay-optimal\\
 \hline
\end{tabular}
\caption{Summary of the delay performance of our proposed policies under different settings, where $d_{\min}$ is the minimum distance among all MDS storage codes defined in \eqref{eq_def2}, $1/\mu$ is the average chunk downloading time of each thread, and $L$ is the number of downloading threads. The classes of ``New-Longer-than-Used'' and ``New-Shorter-than-Used'' distributions are defined in Section \ref{sec4}. Note that the delay gaps in this table are independent of the request arrival process and system traffic load.}\label{tab2}
\end{table*}
\begin{itemize}
\item
When the downloading times of data chunks are \emph{i.i.d.} exponential with mean $1/\mu$, we propose a Shortest Expected Remaining Processing Time policy with Redundant thread assignment (SERPT-R), and prove that SERPT-R is  \emph{delay-optimal} among \emph{all} online policies, if (i) the storage redundancy is sufficiently high and (ii) preemption is allowed. If condition (i) is not satisfied, we show that under SERPT-R, the extra delay caused by low storage redundancy is no more than the average downloading time of $(\ln L +1)$ chunks, i.e., $(\ln L +1)/\mu$, where $L$ is the number of downloading threads. (This delay gap grows slowly with respect to $L$, and is independent of the request arrival process and system traffic load.) Further, if preemption is not allowed, we propose a Shortest Expected Differentiable Processing Time policy with Redundant thread assignment (SEDPT-R), which has a delay gap of
no more than the average downloading time of one chunk, i.e., $1/\mu$, compared to the delay-optimal policy.

\item When the downloading times of data chunks are \emph{i.i.d.} New-Longer-than-Used (NLU) (defined in Section \ref{sec4}), we design a Shortest Expected Differentiable Processing Time policy with Work-Conserving Redundant thread assignment (SEDPT-WCR) for the preemptive case and a Shortest Expected Differentiable Processing Time policy with No Redundant thread assignment (SEDPT-NR) for the non-preemptive case. We show that, comparing with the delay-optimal policy, the delay gaps of preemptive SEDPT-WCR and non-preemptive SEDPT-NR are both of the order $O(\ln{L}/\mu)$. 

\item When the downloading times of data chunks are \emph{i.i.d.} New-Shorter-than-Used (NSU) (defined in Section \ref{sec4}), we prove that SEDPT-R is delay-optimal among all online policies, under the conditions that data is stored with repetition codes, storage redundancy is sufficiently high, and preemption is not allowed. 
    %


\end{itemize}
We note that the proposed SEDPT-type policies are different from the traditional Shortest Remaining Processing Time first (SRPT) policy, and have not been proposed in prior work.

\section{Related Work}

The idea of reducing delay via multiple parallel data transmissions has been explored empirically in various contexts~\cite{DTN-delay,Ananthanarayanan11,vulimiri12latency,vulimiri13latency,Flach13latency}. More recently, theoretical analysis has been conducted to study the delay performance of data retrieval in distributed storage systems. One line of studies \cite{huang-isit-2012,shah-mdsq-2012,Joshi-2012,Xiang2014,Kumar2014,shah-Allerton-2013} were centered on the data retrieval from a small number of storage nodes, where the delay performance is limited by the service capability of individual storage nodes. It was shown in \cite{huang-isit-2012} that erasure storage codes can reduce the queueing delay compared to simple data replications. In \cite{shah-mdsq-2012,Joshi-2012}, delay bounds were provided for First-Come-First-Served (FCFS) policies with different numbers of redundant threads. In \cite{Xiang2014}, a delay upper bound was obtained for FCFS policies under Poisson arrivals and arbitrary downloading time distribution, which was further used to derive a sub-optimal solution for jointly minimizing latency and storage cost. In \cite{Kumar2014}, the authors established delay bounds for the classes of FCFS, preemptive and non-preemptive priority scheduling policies, when the downloading time is \emph{i.i.d.} exponential. In \cite{shah-Allerton-2013}, the authors studied when redundant threads can reduce delay (and when not), and designed optimal redundant thread scheduling policies among the class of FCFS policies.

The second line of researches \cite{addShengboOriginal,Liang2013_2,ShengboInfocom} focus on large-scale storage clouds with a large number of storage nodes, where the delay performance is constrained by the available networking resources of the system. In \cite{addShengboOriginal,Liang2013_2}, the authors measured the chunk downloading time over the Amazon cloud storage system and proposed to adapt code parameters and the number of redundant threads to reduce delay. In \cite{ShengboInfocom}, it was shown that FCFS with redundant thread assignment is delay-optimal among all online policies, under the assumptions of a single storage code, high storage redundancy and exponential downloading time distribution. Following this line of research, in this paper, we consider the more general scenarios with heterogonous storage codes, general level of storage redundancy, and non-exponential downloading time distributions, where neither FCFS nor priority scheduling is close to delay-optimal. 

\section{System Model}
\label{sec2}

We consider a cloud storage system that is composed of one frond-end proxy server and a large number of distributed storage devices, as illustrated in Fig. \ref{fig1model}. The proxy server enqueues the user requests and establishes TCP connections to fetch data from the storage devices. In practice, the proxy server also performs tasks such as format conversion, data compression, authentication and encryption.\footnote{Our results can be also used for systems with multiple proxy servers, where each read request is routed to a proxy server based on geometrical location, or determined by a round robin or random load balancing algorithm. More complicated load balancing algorithms will be studied in our future work.}

\subsection{Data Storage and Retrieval}
Suppose that the file corresponding to request $i$ is stored with an $(n_i,k_i)$ MDS code.\footnote{The terms ``file'' and ``request'' are interchangeable in this paper.} Then, file $i$ is partitioned into $k_{i}$ equal-size chunks, which are encoded into $n_{i}$ coded chunks and stored in $n_{i}$ distinct devices. In MDS codes, any $k_i$ out of the $n_i$ coded chunks are sufficient to restore file $i$. Therefore, the cloud storage system can tolerate $n_i-k_i$ failures and still secure file $i$. Examples of MDS codes include repetition codes ($k_i=1$) and Reed-Solomon codes. Let $d_i$ denote the Hamming distance of an $(n_i,k_i)$ MDS code, determined by
\begin{eqnarray}\label{eq_d_i}
d_{i} = n_i-k_i+1.
\end{eqnarray}
The minimum code distance of all storage codes is defined as
\begin{eqnarray}\label{eq_def2}
d_{\min} \triangleq \min \{d_{i},i=1,2,\cdots\}.
\end{eqnarray}

It has been reported in \cite{Dean-CACM13,Garfinkel07anevaluation,Wang:2010,Google2003,ErasureCodingWindows} that the downloading time of data chunks can be highly unpredictable in storage clouds. Some recent measurements \cite{addShengboOriginal,Liang2013_2,ShengboInfocom} on Amazon AWS show that the downloading times of data chunks stored with distinct keys can be approximated as independent and identically distributed (\emph{i.i.d.}) random variables. In this paper, we assume that the downloading times of data chunks are \emph{i.i.d.}\footnote{This assumption is reasonable for large-scale storage clouds, e.g., Amazon AWS, where individual read operations may experience long latency events, such as network congestion, cache misses, database blocking, high temperature or high I/O traffic of storage devices, that are unobservable and unpredictable by the decision-maker.},
as in \cite{huang-isit-2012,shah-Allerton-2013,shah-mdsq-2012,Joshi-2012,ShengboInfocom}.

\begin{figure}
\centering \includegraphics[width=3in]{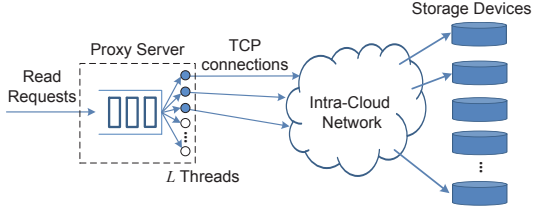} \caption{System model.}
\label{fig1model} \vspace{-0.5cm}
\end{figure}

\subsection{Redundant and Parallel Thread Scheduling}
The proxy server has $L$ downloading threads, each representing a potential TCP connection, to retrieve data from the distributed storage devices. The value of $L$ is chosen as the maximum number of simultaneous TCP connections that can occupy all the available networking bandwidth without significantly degrading the latency of each individual connection \cite{addShengboOriginal,Liang2013_2}. A decision-maker at the proxy server determines which chunks to download and in what order for the $L$ threads to minimize the average data retrieving delay.

Suppose that a sequence of $N$ read requests arrive at the queue of the processing server.\footnote{The value of $N$ can be either finite or infinite in this paper. If $N$ tends to infinite, a $\limsup$ operator is enforced on the right hand side of \eqref{eq_2}.} Let $a_i$ and $c_{i,\pi}$ denote the arrival and completion times of the $i$th request under policy $\pi$, respectively, where $0=a_1\leq a_2\leq \cdots \leq a_N$. Thus, the service latency of request $i$ is given by
$c_{i,\pi}-a_i$,
which includes both the downloading time and the waiting time in the request queue.
We assume that the arrival process $(a_1,a_2,\cdots)$ is an \emph{arbitrary deterministic} time sequence, while the departure process $(c_{1,\pi},c_{2,\pi},\cdots)$ is stochastic because of the random downloading time. Given the request parameters $N$ and $(a_i,k_i,n_i)_{i=1}^N$, the average flow time of the  requests under policy $\pi$ is defined as
\begin{align}
\overline{D}_\pi=\frac{1}{N}\sum_{i=1}^{N}\left(\mathbb{E}\left\{c_{i,\pi}\right\}-a_i\right),\label{eq_2}
\end{align}
where the expectation is taken with respect to the random distribution of chunk downloading time for given policy $\pi$ and for given request parameters $N$ and $(a_i,k_i,n_i)_{i=1}^N$.
The goal of this paper is {\it to design low-complexity online thread scheduling policies that achieve optimal or near-optimal delay performance.}

\begin{definition}
\textbf{Online policy}: A scheduling policy is said to be \emph{online} if, at any given time $t$, the decision-maker does not know the number of requests to arrive after time $t$, the parameters $(a_i,k_i,n_i)$ of the requests to arrive, or the realizations of the (remaining) downloading times of the chunks that have not been accomplished by time $t$. 
\end{definition}

\begin{definition}
\textbf{Delay-optimality}: A thread scheduling policy $\pi$ is said to be \emph{delay-optimal} if, for {\it any} given request parameters
$N$ and $(a_i,k_i,n_i)_{i=1}^N$, it yields the shortest average flow time $\overline{D}_\pi$ among {\it all online} policies.
\end{definition}

A key feature of this scheduling problem is the flexibility of \emph{redundant and parallel thread scheduling}. 
Take file $i$ as an example. When $n_i>k_i$, one can assign redundant threads to download more than $k_i$ chunks of file $i$. The
first $k_i$ successfully downloaded chunks are sufficient for completing the read operation. After that, the extra downloading threads are terminated immediately, which is called \emph{service termination}. By doing this, the retrieving delay of file $i$ is reduced. On the other hand, redundant thread scheduling may cause extra system load. Therefore, such a policy provides a tradeoff between fast retrieving of each file and a potentially longer service latency due to the extra system load, which makes it difficult to achieve delay-optimality.
\subsection{Service Preemption and Work Conserving}
We consider {\it chunk-level} preemptive and non-preemptive policies.
When preemption is allowed, a thread can switch to serve another chunk at any time, and resume to serve the previous chunk at a later time, continuing from the interrupted point.
When preemption is not allowed, a thread must complete (or terminate) the current chunk before switching to serve another chunk.
We assume that service terminations and  preemptions are executed immediately with no extra delay. 

\begin{definition}\label{def_5} \textbf{Work-conserving:}
A scheduling policy is said to be \emph{work-conserving} if all threads are kept busy
whenever there are chunks waiting to be downloaded.
\end{definition}

{\bf Remark 1:} \emph{If preemption is allowed, a delay-optimal policy must be work-conserving}, because the average delay of any non-work-conserving policy can be reduced by assigning the idle threads to download more chunks. Meanwhile, \emph{if preemption is not allowed, a work-conserving policy may not be delay-optimal}, because the occupied threads cannot be easily switched to serve an incoming request with a higher priority.

%
%


\ignore{
\subsubsection{Available Chunks and Outstanding Chunks}
{Therefore, it cannot predict beforehand when and which chunks will be downloaded. However, after some chunks are downloaded, the decision-maker knows about how many additional chunks are needed for retrieving each file.}

\begin{definition}
\textbf{Available chunks}: the chunks that have not been downloaded for each file.
\end{definition}

\begin{definition}
\textbf{Outstanding chunks}: from the set of available chunks, any subset of chunks that are sufficient for retrieving each file.
\end{definition}

Let $u_{i}(t)$ denote the number of available chunks of file $i$ at time $t$, and $r_i(t)$ denote the number of outstanding chunks of file $i$ at time $t$. The parameters $(u_i(t),r_i(t))$ will be used to make scheduling decisions.

Suppose that $g_i(t)$ chunks of file $i$ have been downloaded by time $t$. If $g_i(t)< k_i$, then require $i$ is not completed at time $t$, such that $r_i(t)=k_i-g_i(t)$ and $u_i(t)=n_i-g_i(t)$.
Therefore, any unfinished request $i$ satisfies
\begin{eqnarray}\label{eq_u}
u_i(t)-r_{i}(t) = n_i-k_i = d_i-1.
\end{eqnarray}
If $g_i(t)= k_i$, then request $i$ is completed before time $t$ such that $r_i(t)=u_i(t)=0$.
We further define
\begin{eqnarray}\label{eq_def4}
r(t)=\sum_{i:a_i\leq t} r_i(t)
\end{eqnarray}
as the total number of outstanding chunks of all unfinished requests at time $t$, and
\begin{eqnarray}\label{eq_def3}
u(t)=\sum_{i:a_i\leq t} u_i(t)
\end{eqnarray}
as the total number of available chunks of all unfinished requests at time $t$.
}

\section{Exponential Chunk Downloading Time}

\label{sec3}
In this section, we study the delay-optimal thread scheduling when chunk downloading time is \emph{i.i.d.} exponentially distributed with {mean $1/\mu$. Non-exponential downloading time distributions will be investigated in Section \ref{sec4}.}


\subsection{High Storage Redundancy, Preemption is Allowed}\label{sec3-1}
{We first consider the case of high storage redundancy such that $d_{\min} \geq L$ is satisfied. In this case, we have $n_i-(k_i -1) \geq L$ for all $i$. Hence, each file $i$ has at least $L$ available chunks even if $k_i-1$ chunks of file $i$ have been downloaded. Hence, each unfinished request has sufficient available chunks such that all $L$ threads can be simultaneously assigned to serve this request.}
Let $s_j$ denote the arrival time of the $j$th arrived chunk downloading task of all files and $t_j$ denote the completion time of the $j$th downloaded chunk of all files. The chunk arrival process $(s_1,s_2,\ldots)$ is uniquely determined by the request parameters $(a_i,k_i)_{i=1}^N$. Meanwhile, the chunk departure process $(t_1,t_2,\ldots)$ satisfies the following {\it invariant distribution} property:

\begin{lemma}\cite[Theorem 6.4]{ShengboInfocom}\label{lem0-1} Suppose that (i) $d_{\min} \geq L$ and (ii) the chunk downloading time is {i.i.d.} exponentially distributed with mean $1/\mu$. Then, for any given request parameters $N$ and $(a_i,k_i,n_i)_{i=1}^N$, the distribution of the chunk departure process $(t_1, t_2, \ldots)$ is invariant under any work-conserving policy. \end{lemma}

%
%
We propose a preemptive Shortest Expected Remaining Processing Time first policy with Redundant thread assignment (\textbf{preemptive SERPT-R}):

\vspace{0.1cm}


\emph{
Suppose that, at any time $t$, there are $V$ unfinished requests $i_1,i_2,\ldots, i_V$, such that $\alpha_j$ chunks need to be downloaded for completing request $i_j$. Under SERPT-R, each idle thread is assigned to serve one available chunk of request $i_j$ with the smallest $\alpha_j$. (Due to storage redundancy, the number of available chunks of request $i_j$ is larger than $\alpha_j$.) If all the available chunks of request $i_j$ are under service, then the idle thread is assigned to serve one available chunk of request $i_{j'}$ with the second smallest $\alpha_{j'}$. This procedure goes on, until all $L$ threads are occupied or all the available chunks of the $V$ unfinished requests are under services.
}
\vspace{0.1cm}

This policy is an extension of Shortest Remaining Processing Time first (SRPT) policy \cite{Schrage68,Smith78} because it schedules parallel and redundant downloading threads to serve the requests with the least workload.
The following theorem shows that this policy is delay-optimal under certain conditions.



\begin{Theorem}\label{thm1}
Suppose that (i) $d_{\min} \geq L$, 
(ii) preemption is allowed,
and (iii) the chunk downloading time is {i.i.d.} exponentially distributed with mean $1/\mu$. Then, for any given request parameters $N$ and $(a_i,k_i,n_i)_{i=1}^N$, preemptive SERPT-R is delay-optimal among all online policies.
\end{Theorem}


{\bf Remark 2:} Theorem \ref{thm1} and the subsequent theoretical results of this paper 
are difficult to establish for the following reasons: 1) Each request $i$ is partitioned into a batch of $k_i$ chunk downloading tasks, and the processing time of each task is random. 2) There are $n_i-k_i$ redundant chunks for request $i$, such that completing any $k_i$ of the $n_i$ tasks would complete the request. 3) The system has $L$ threads which can simultaneously process $L$ tasks belonging to one or multiple requests. 4) If redundant downloading threads are scheduled, the associated extra system load must be  considered when evaluating the delay performance.

\begin{proof}
We provide a proof sketch of Theorem \ref{thm1}. Consider an arbitrarily given chunk departure sample path $(t_1,t_2,\ldots)$. According to the property of the SRPT principle \cite{Schrage68,Smith78}, preemptive SERPT-R minimizes $\frac{1}{N}\sum_{i=1}^{N}\left(c_{i,\pi}-a_i\right)$ for any given sample path $(t_1,t_2,\ldots)$. Further, Lemma \ref{lem0-1} tells us that the distribution of $(t_1,t_2,\ldots)$ is invariant among the class of work-conserving policies. By this, preemptive SERPT-R is delay-optimal among the class of work-conserving policies. Finally, since a delay-optimal policy must be work-conserving when preemption is allowed, Theorem \ref{thm1} follows.
\ifreport
More details are provided in Appendix \ref{app0}.
\else
See \cite{tech_report2014} for the details.
\fi
\end{proof}



%

In Theorem 6.4 of \cite{ShengboInfocom}, it was shown that a First-Come-First-Served policy with Redundant thread assignment (FCFS-R) is delay-optimal when $k_{i}=k$ for all $i$ and $d_{\min}\geq L$. In this case, preemptive SERPT-R reduces to the following policy: {After a request departs from the system, pick any waiting request (not necessarily the request arrived the earliest) and assign all $L$ threads to serve the available chunks of
this request until it departs}. Hence, FCFS-R belongs to the class of SERPT-R policies, and Theorem 6.4 of \cite{ShengboInfocom} is a special case of Theorem \ref{thm1}.


\subsection{General Storage Redundancy, Preemption is Allowed}
When $d_{\min}< L$, some requests may have less than $L$ available chunks, such that not all of the $L$ threads can be assigned to serve it.
In this case, SERPT-R may not be delay-optimal. This is illustrated in the following example.

\begin{example}\label{ex1}
%
Consider two requests with parameters given as $(k_1 = 1,n_1 = 4, d_1 =4, a_1 = 0)$ and $(k_2 = 2,n_2 = 2, d_2 =1, a_2 = 0)$. The number of threads is $L=4$. Under SERPT-R, all $4$ threads are assigned to serve request $1$ after time zero. However, after request $1$ is completed, the chunk downloading rate is reduced from $4\mu$ to $2\mu$, because request $2$ only has $n_2 = 2$ chunks. Furthermore, after one chunk of request $2$ is downloaded, the chunk downloading rate is reduced from $2\mu$ to $\mu$. The average flow time of SERPT-R is $\overline{D}_{\text{SERPT-R }}=1/\mu$ seconds.

We consider another policy $Q$: after time zero, $2$ threads are assigned to serve request 1 and $2$ threads are assigned to serve request $2$. After the first chunk is downloaded, if the downloaded chunk belongs to request $1$, then request 1 departs and 2 threads are assigned to serve request $2$. If the downloaded chunk belongs to request $2$, then 3 threads are assigned to serve request $1$ and $1$ thread is assigned to serve request $2$.
After the second chunk is downloaded, only one request is left and the threads are assigned to serve the available chunks of this request. The average flow time of policy $Q$ is $\overline{D}_{Q}=61/(64\mu)$ seconds. Hence, SERPT-R is not delay-optimal.
\end{example}


Next, we bound the delay penalty associated with removing the condition $d_{\min}\geq L$.
\begin{Theorem}\label{thm3}
If (i) preemption is allowed and (ii) the chunk downloading time is {i.i.d.} exponentially distributed with mean $1/\mu$. Then, for any given request parameters $N$ and $(a_i,k_i,n_i)_{i=1}^N$, the average flow time of preemptive SERPT-R satisfies
\begin{eqnarray}
&&\!\!\!\!\!\!\!\!\!\!\!\!\!\!\!\!\!\!\!\!\!\overline{D}_{\text{opt}}\leq \overline{D}_{\text{prmp,SERPT-R}}
\leq\overline{D}_{\text{opt}}+
\frac{1}{\mu}\sum_{l=d_{\min}}^{L-1}\frac{1}{l},\label{eq_4-1-1}
\end{eqnarray}
where $d_{\min}$ is defined in \eqref{eq_def2}.
\end{Theorem}

\begin{proof}
Here is a proof sketch of Theorem \ref{thm3}. We first use a {\it state evolution} argument to show that, after removing the condition $d_{\min}\geq L$, SERPT-R needs to download $L-d_{\min}$ or fewer additional chunks after any time $t$, so as to accomplish the same number of requests that are completed by SERPT-R with the condition $d_{\min}\geq L$ during $(0,t]$. Further, according to the properties of exponential distribution, the average time for the system to download $L-d_{\min}$ extra chunks under the conditions of Theorem \ref{thm3} is upper bounded by the last term of \eqref{eq_4-1-1}. This completes the proof.
\ifreport
See Appendix \ref{app1} for the details.
\else
See \cite{tech_report2014} for the details.
\fi
\end{proof}
Note that if $d_{\min}\geq L$, the last term in \eqref{eq_4-1-1} becomes zero which corresponds to the case of Theorem \ref{thm1}; if $d_{\min}< L$, the last term in \eqref{eq_4-1-1} is upper bounded by $\frac{1}{\mu}\left[\ln(\frac{L-1}{d_{\min}})+1\right]$. Therefore, the delay penalty caused by low storage redundancy is of the order $O(\ln L/\mu)$, and is insensitive to increasing $L$. Further, this delay penalty remains constant for any request arrival process and system traffic load.



\subsection{High Storage Redundancy, Preemption is Not Allowed}
Under preemptive SERPT-R, each thread can switch to serve another request at any time. However, when preemption is not allowed, a thread must complete or terminate the current chunk downloading task before switching to serve another request. In this case, SERPT-R may not be delay-optimal, as illustrated in the following example.

\begin{example}\label{ex2}
Consider two requests with parameters given as $(k_1 = 2,n_1 = 3,d_1=2, a_1 = 0)$ and $(k_2 = 1,n_2 = 2,d_2=2, a_2 = \varepsilon)$, where $\varepsilon>0$ can be arbitrarily close to zero. The number of threads is $L=2$, the chunk downloading time is {i.i.d.} exponentially distributed with mean $1/\mu$. Under SERPT-R, the two threads are assigned to serve request $1$ after time zero. After the first chunk is downloaded, one thread is assigned to serve request $2$ and the other thread remains to serve request $1$. After the second chunk is downloaded, one of the requests has departed, and the two threads are assigned to serve the remaining request.
The average flow time of SERPT-R is $\overline{D}_{\text{SERPT-R}}=5/(4\mu)-\varepsilon/2$ seconds.

We consider another non-preemptive policy $Q$: the threads remain idle until time $\varepsilon$. After $\varepsilon$, the two threads are assigned to serve request $2$. After the first chunk is downloaded, request $2$ has departed. Then, the two threads are assigned to serve request $1$, until it departs. The average flow time of policy $Q$ is $\overline{D}_{Q}=1/\mu+\varepsilon/2$ seconds. Since $\varepsilon$ is arbitrarily small, SERPT-R is not delay-optimal when preemption is not allowed.
\end{example}


We propose a non-preemptive Shortest Expected Differential Processing Time first policy with Redundant thread assignment (\textbf{non-preemptive SEDPT-R}), where the service priority of a file is determined by the difference between the number of remaining chunks of the file and the number of threads that has been assigned to the file.

\vspace{0.1cm}
\emph{
Suppose that, at any time $t$, there are $V$ unfinished requests $i_1,i_2,\ldots, i_V$, such that $\alpha_j$ chunks need to be downloaded for completing request $i_j$ at time $t$ and $\delta_j$ threads have been assigned to serve request $i_j$. Under non-preemptive SEDPT-R, each idle thread is assigned to serve one available chunk of request $i_j$ with the smallest $\alpha_j-\delta_j$. (Due to storage redundancy, the number of available chunks of request $i_j$ is larger than $\alpha_j$. Hence, it may happen that $\alpha_j-\delta_j<0$ because of redundant chunk downloading.) If all the available chunks of request $i_j$ are under service, then the idle thread is assigned to serve one available chunk of request $i_{j'}$ with the second smallest $\alpha_{j'}-\delta_{j'}$. This procedure goes on, until all $L$ threads are occupied or all the available chunks of the $V$ unfinished requests are under services.
}
\vspace{0.1cm}




{The intuition behind non-preemptive SEDPT-R is that $\delta_j$ chunks of request $i_j$ will be under service after time $t$ for any non-preemptive policy, and thus should be excluded when determining the service priority of request $i_j$. This is different from the traditional SRPT-type policies \cite{Schrage68,Smith78,Michael2012book,Fox:2011:OSI:2133036.2133046,doi:10.1137/090772228}, which do not exclude the chunks under service when determining the service priorities of the requests.}
The delay performance of this policy is characterized in the following theorem: 
\begin{Theorem}\label{thm2}
Suppose that (i) $d_{\min}\geq L$, 
(ii) preemption is \textbf{not} allowed, and (iii) the chunk downloading time is {i.i.d.} exponentially distributed with mean $1/\mu$. Then, for any given request parameters $N$ and $(a_i,k_i,n_i)_{i=1}^N$, the average flow time of non-preemptive SEDPT-R satisfies
\begin{eqnarray} \label{eq_2-1}
\overline{D}_{\text{opt}}\leq\overline{D}_{\text{non-prmp,SEDPT-R}}
\leq\overline{D}_{\text{opt}}\!+1/\mu. 
\end{eqnarray}
\end{Theorem}

\begin{proof}
We provide a proof sketch of Theorem \ref{thm2}. Theorem \ref{thm1} tells us that preemptive SERPT-R provides a lower bound of $\overline{D}_{\text{opt}}$. On the other hand, non-preemptive SEDPT-R provides an upper bound of $\overline{D}_{\text{opt}}$.
Thus, we need to show that the delay gap between preemptive SERPT-R and non-preemptive SEDPT-R is at most $1/\mu$. Towards this goal, we use a {\it state evolution} argument to show that {for any time $t$ and any given sample path of chunk departures $(t_1,t_2,\ldots)$, non-preemptive SEDPT-R needs to download $L$ or fewer additional chunks after time $t$, so as to accomplish the same number of requests that are completed under preemptive SERPT-R during $(0,t]$}. By the properties of exponential distribution, the average time for the $L$ threads to download $L$ chunks under non-preemptive SEDPT-R is $1/\mu$, and Theorem \ref{thm2} follows.
\ifreport
See Appendix \ref{app2} for the details.
\else
See \cite{tech_report2014} for the details.
\fi
\end{proof}





Theorem \ref{thm2} tells us that the delay gap between non-preemptive SEDPT-R and the optimal policy is at most the average downloading time of one chunk by each thread, i.e., $1/\mu$. Intuitively speaking, this is because each thread only needs to wait for downloading one chunk, before switching to serve another request. However, the proof of Theorem \ref{thm2} is non-trivial, because it must work for any possible sample path of the downloading procedure.

\subsection{General Storage Redundancy, Preemption is Not Allowed}

When preemption is {not} allowed and the condition $d_{\min}\geq L$ is removed, we have the following result.
\begin{Theorem}\label{thm4}
Suppose that (i) preemption is \textbf{not} allowed, and (ii) the chunk downloading time is {i.i.d.} exponentially distributed with mean $1/\mu$. Then, for any given request parameters $N$ and $(a_i,k_i,n_i)_{i=1}^N$, the average flow time of non-preemptive SEDPT-R satisfies
\begin{eqnarray}\label{coro2-1-1}
&&\!\!\!\!\!\!\!\!\!\!\!\!\!\!\!\!\!\!\!\!\!\!\overline{D}_{\text{opt}}\leq\overline{D}_{\text{non-prmp,SEDPT-R}}
\leq\overline{D}_{\text{opt}}\!+\frac{1}{\mu}+\frac{1}{\mu}\sum_{l=d_{\min}}^{L-1}\frac{1}{l},
\end{eqnarray}
where $d_{\min}$ is defined in \eqref{eq_def2}.
\end{Theorem}
\ifreport
\begin{proof}
See Appendix \ref{app3-1}.
\end{proof}
\else

\fi

If $d_{\min}\geq L$, the last term in \eqref{coro2-1-1} becomes zero which corresponds to the case of Theorem \ref{thm2}.

\section{Non-Exponential Chunk Downloading Time}\label{sec4}
In this section, we consider two classes of general downloading time distributions: 
New-Longer-than-Used (NLU) distributions and New-Shorter-than-Used (NSU) distributions, defined as follows.\footnote{Note that New-Longer-than-Used (New-Shorter-than-Used) is \textbf{equivalent} to the term New-Better-than-Used (New-Worse-than-Used) used in reliability theory \cite{Bagnoli2005,StochasticOrderBook}, where ``better'' means a longer lifetime. However, this may lead to confusion in the current paper, where ``better'' means a shorter delay. We choose to use New-Longer-than-Used (New-Shorter-than-Used) to avoid confusion. In a recent work \cite{shah-Allerton-2013}, the New-Longer-than-Used (New-Shorter-than-Used) property was termed light-everywhere (heavy-everywhere).}

\begin{definition}
\textbf{New-Longer-than-Used distributions}: A distribution on $[0,\infty)$ is said to be New-Longer-than-Used (NLU), if for all $t,\tau\geq 0$ and $\mathbb{P}(X>\tau)>0$, the distribution satisfies
\begin{eqnarray}\label{eq_light}
\mathbb{P}(X>t) \geq \mathbb{P}(X>t+\tau|X>\tau).
\end{eqnarray}
%
%
\textbf{New-Shorter-than-Used distributions}: A distribution on $[0,\infty)$ is said to be New-Shorter-than-Used (NSU), if for all $t,\tau\geq 0$ and $\mathbb{P}(X>\tau)>0$, the distribution satisfies
\begin{eqnarray}\label{eq_heavy}
\mathbb{P}(X>t) \leq \mathbb{P}(X>t+\tau|X>\tau).
\end{eqnarray}
\end{definition}

NLU (NSU) distributions are closely related to log-concave (log-convex) distributions.
Many commonly used distributions are NLU or NSU distributions \cite{Bagnoli2005}. In practice, NLU distributions can be used to characterize the scenarios where the downloading time is a constant value followed by a short latency tail. For instance, recent studies \cite{addShengboOriginal,Liang2013_2} suggest that the data downloading time of Amazon AWS can be approximated as a constant delay plus an exponentially distributed random variable, which is an NLU distribution. On the other hand, NSU distributions can be used to characterize occasional slow responses resulting from TCP retransmissions, I/O interference, database blocking and/or even disk failures.

We will require the following definitions: Let $\vec{x} = (x_1, x_2,\ldots,x_m)$ and $\vec{y} = (y_1, y_2,\ldots,y_m)$ be two vectors in $\mathbb{R}^m$, then we denote $\vec{x} \leq \vec{y}$ if $x_i \leq y_i$ for $i = 1,2,\ldots,m$.
\begin{definition}
\textbf{Stochastic Ordering}: \cite{StochasticOrderBook}
Let ${X}$ and ${Y}$ be two random variables. Then, ${X}$ is said to be {stochastically smaller than ${Y}$} (denoted as ${X}\leq_{\text{st}}{Y}$), if
\begin{eqnarray}
\mathbb{P}({X}>t) \leq \mathbb{P}({Y}>t)~\text{for all }t\in \mathbb{R}.
\end{eqnarray}
\end{definition}

\begin{definition}
\textbf{Multivariate Stochastic Ordering}: \cite{StochasticOrderBook}
A set $U \subseteq \mathbb{R}^m$ is called \emph{upper} if $\vec{y} \in U$ whenever $\vec{y}\geq \vec{x}$ and $\vec{x} \in U$.
Let $\vec{X}$ and $\vec{Y}$ be two random vectors. Then, $\vec{X}$ is said to be {stochastically smaller than $\vec{Y}$} (denoted as $\vec{X}\leq_{\text{st}}\vec{Y}$), if
\begin{eqnarray}
\mathbb{P}(\vec{X}\in U) \leq \mathbb{P}(\vec{Y}\in U)~\text{for all upper sets}~U\subseteq \mathbb{R}^m.
\end{eqnarray}
\end{definition}
Stochastic ordering of stochastic processes (or infinite vectors) can be defined similarly \cite{StochasticOrderBook}.

\subsection{NLU Chunk Downloading Time Distributions}\label{sec4:light}

%
%
%

We consider a non-preemptive Shortest Expected Differential Processing Time first policy with No Redundant thread assignment (\textbf{non-preemptive SEDPT-NR}):

\vspace{0.1cm}

\emph{
Suppose that, at any time $t$, there are $V$ unfinished requests $i_1,i_2,\ldots, i_V$, such that $\alpha_j$ chunks need to be downloaded for completing request $i_j$  at time $t$ and $\delta_j$ threads have been assigned to serve request $i_j$. Under non-preemptive SEDPT-NR, each idle thread is assigned to serve one available chunk of request $i_j$ with the smallest $\alpha_j-\delta_j$. If $\alpha_j$ threads have been assigned to request $i_j$, then the idle thread is assigned to serve one available chunk of request $i_{j'}$ with the second smallest $\alpha_{j'}-\delta_{j'}$. This procedure goes on, until all $L$ threads are occupied or each request $i_j$ is served by $\alpha_j$ threads.
}
\vspace{0.1cm}

Note that since at most $\alpha_j$ threads are assigned to request $i_j$, we have $\alpha_j-\delta_j\geq0$ for all $i_j$ under non-preemptive SEDPT-NR.
SEDPT-NR is a non-work-conserving policy. When preemption is allowed, the delay performance of SEDPT-NR can be improved by exploiting the idle threads to download redundant chunks. This leads to a preemptive Shortest Expected Differential Processing Time first policy with Work-Conserving Redundant thread assignment (\textbf{preemptive SEDPT-WCR}):

\vspace{0.1cm}
\emph{Upon the decision of {SEDPT-NR}, if each request $i_j$ is served by $\alpha_j$ threads and there are still some idle threads, then assign these threads to download some redundant chunks to avoid idleness. When a new request arrives, the threads downloading redundant chunks will be preempted to serve the new arrival request.}
\vspace{0.1cm}

Let us consider the service time for a thread to complete downloading one chunk. If the thread has spent $\tau$ seconds on one chunk, the tail probability for completing the current chunk under service is $\mathbb{P}(X>t+\tau|X>\tau)$. On the other hand, the tail probability for switching to serve a new chunk is $\mathbb{P}(X>t)$. Since the chunk downloading time is \emph{i.i.d.} NLU, it is stochastically better to keep downloading the same chunk than switching to serve a new chunk. 

\begin{lemma}\label{lem2} Suppose that (i) the system load is high such that all $L$ threads are occupied at all time $t\geq0$ and (ii) the chunk downloading time is {i.i.d.} NLU. Then, for any given request parameters $N$ and $(a_i,k_i,n_i)_{i=1}^N$, the chunk departure instants $(t_1,t_2,\ldots)$ under non-preemptive SEDPT-NR are stochastically smaller than those under any other online policy.
\end{lemma}

\ifreport
\begin{proof}
See Appendix \ref{app3}.
\end{proof}
\else
\fi

\begin{lemma}\label{thm6}
Suppose that (i) the system load is high such that all $L$ threads are occupied at all time $t\geq0$, (ii) preemption is {\bf not} allowed, and (iii) the chunk downloading time is {i.i.d.} NLU. Then, for any given request parameters $N$ and $(a_i,k_i,n_i)_{i=1}^N$, the average flow time of non-preemptive SEDPT-NR satisfies
\begin{eqnarray}\label{eq_6}
&&\!\!\!\!\!\!\!\!\!\!\!\!\!\!\!\!\!\!\!\!\!\overline{D}_{\text{opt}}\leq \overline{D}_{\text{non-prmp,SEDPT-NR}} \leq\overline{D}_{\text{opt}}+
\mathbb{E}\left\{ \max_{l=1,\ldots, L} X_l\right\}, 
\end{eqnarray}
where the $X_l$'s are {i.i.d.} chunk downloading times.
\end{lemma}

\ifreport
\begin{proof}
See Appendix \ref{app4}.
\end{proof}
\else
\fi

If the average chunk downloading time is $\mathbb{E}\left\{ X_l\right\}=1/\mu$, then the last term in \eqref{eq_6} is bounded by
\begin{eqnarray}\label{eq_remark7}
\frac{1}{\mu}\leq \mathbb{E}\left\{ \max_{l=1,\ldots, L} X_l\right\} \leq \frac{1}{\mu}\sum_{l=1}^{L}\frac{1}{l},
\end{eqnarray}
where the lower bound is trivial, and the upper bound follows from the property of New-Longer-than-Used distributions in Proposition 2 of \cite{Jose}. Therefore, the delay gap in Lemma \ref{thm6} is no more than $(\ln L + 1)/\mu$.
Next, we remove condition (i) in Lemma \ref{thm6} and obtain the following result.
\begin{Theorem}\label{thm7}
Suppose that 
(i) preemption is {\bf not} allowed and (ii) the chunk downloading time is {i.i.d.} NLU. Then, for any given request parameters $N$ and $(a_i,k_i,n_i)_{i=1}^N$, the average flow time of non-preemptive SEDPT-NR satisfies
\begin{eqnarray}\label{eq_7}
&&\!\!\!\!\!\!\!\!\!\!\!\!\!\!\!\overline{D}_{\text{opt}}\leq \overline{D}_{\text{non-prmp,SEDPT-NR}} \leq\overline{D}_{\text{opt}}\nonumber\\
&&\!\!\!\!\!\!+\mathbb{E}\left\{ \max_{l=1,\ldots, L} X_l\right\} + \mathbb{E}\left\{ \max_{l=1,\ldots, L-1} X_l\right\},~
\end{eqnarray}
where the $X_l$'s are {i.i.d.} chunk downloading times.
\end{Theorem}

\ifreport
\begin{proof}
See Appendix \ref{app5}.
\end{proof}
\else
\fi

When preemption is allowed, preemptive SEDPT-WCR can achieve a shorter average delay than non-preemptive SEDPT-NR. In this case, we have the following result.
\begin{Theorem}\label{thm8}
Suppose that 
(i) preemption is allowed and (ii) the chunk downloading time is {i.i.d.} NLU. Then, for any given request parameters $N$ and $(a_i,k_i,n_i)_{i=1}^N$, the average flow time of preemptive SEDPT-WCR satisfies
\begin{eqnarray}
&&\!\!\!\!\!\!\!\!\!\!\!\!\!\!\!\overline{D}_{\text{opt}}\leq \overline{D}_{\text{prmt,SEDPT-WCR}} \leq\overline{D}_{\text{opt}}\nonumber\\
&&\!\!\!\!\!\!+\mathbb{E}\left\{ \max_{l=1,\ldots, L} X_l\right\} + \mathbb{E}\left\{ \max_{l=1,\ldots, L-1} X_l\right\},~
\end{eqnarray}
where the $X_l$'s are {i.i.d.} chunk downloading times.
\end{Theorem}
\ifreport
\begin{proof}
See Appendix \ref{app5.5}.
\end{proof}
\else
\fi

Similar to Lemma \ref{thm6}, the delay gaps in Theorems \ref{thm7} and \ref{thm8} are also of the order $O(\ln L/\mu)$.



\subsection{NSU Chunk Downloading Time Distributions}\label{sec4:heavy}
If the chunk downloading time is \emph{i.i.d.} NSU, one can show that it is stochastically better to switch to a new chunk than sticking to downloading the same chunk.
We consider the scenario that preemption is not allowed and obtain the following result.
\begin{lemma}\label{lem3} Suppose that (i) $d_{\min}\geq L$,  (ii) $k_i = 1$ for all $i$,
(iii) preemption is {not} allowed, and (iv) the chunk downloading time is {i.i.d.} NSU. Then, for any given request parameters $N$ and $(a_i,k_i=1,n_i)_{i=1}^N$, the chunk departure instants $(t_1,t_2,\ldots, t_N)$  under non-preemptive SEDPT-R are stochastically smaller than those under any other online policy.
\end{lemma}

\begin{figure*}[!t]
\begin{center}
{\subfigure[][Preemption is allowed, $d_{\min}=L=3$]{\resizebox{0.3\textwidth}{!}{\includegraphics{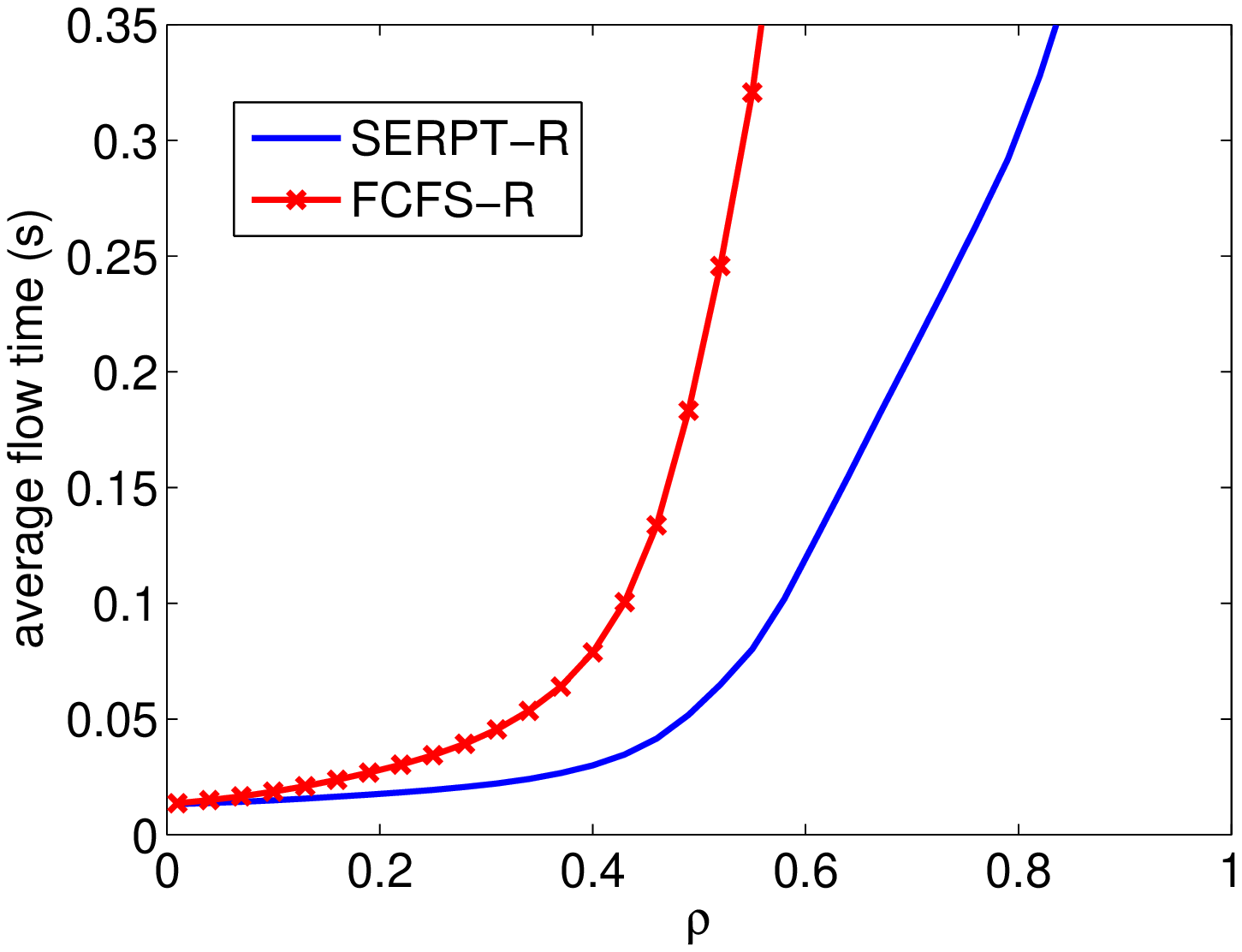}}}}
    \hspace{0.1\textwidth}
{\subfigure[][Preemption is not allowed, $d_{\min}=L=3$]{\resizebox{0.3\textwidth}{!}{\includegraphics{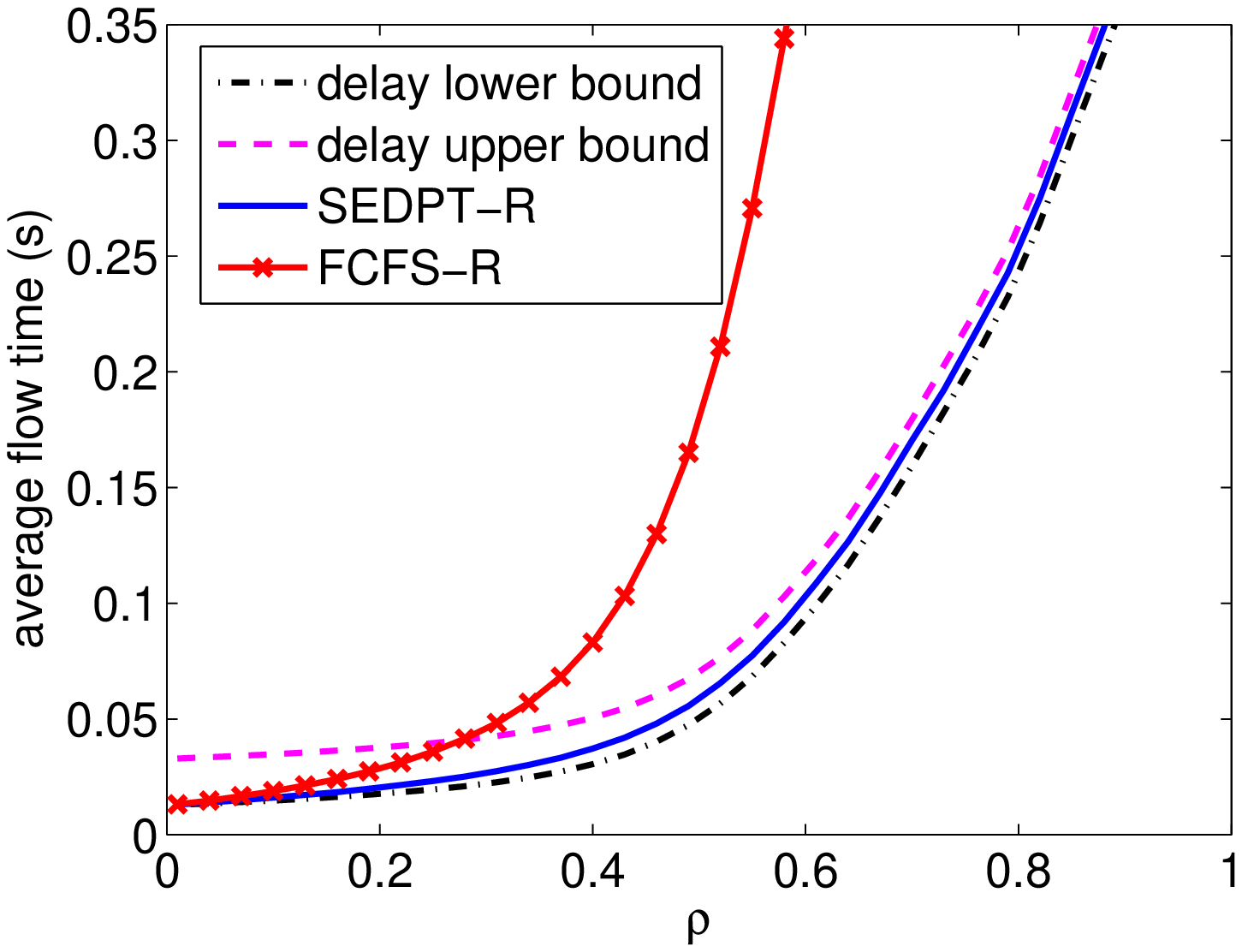}}}
    \vspace{0.0cm}}
{\subfigure[][Preemption is allowed, $d_{\min}<L=5$]{\resizebox{0.3\textwidth}{!}{\includegraphics{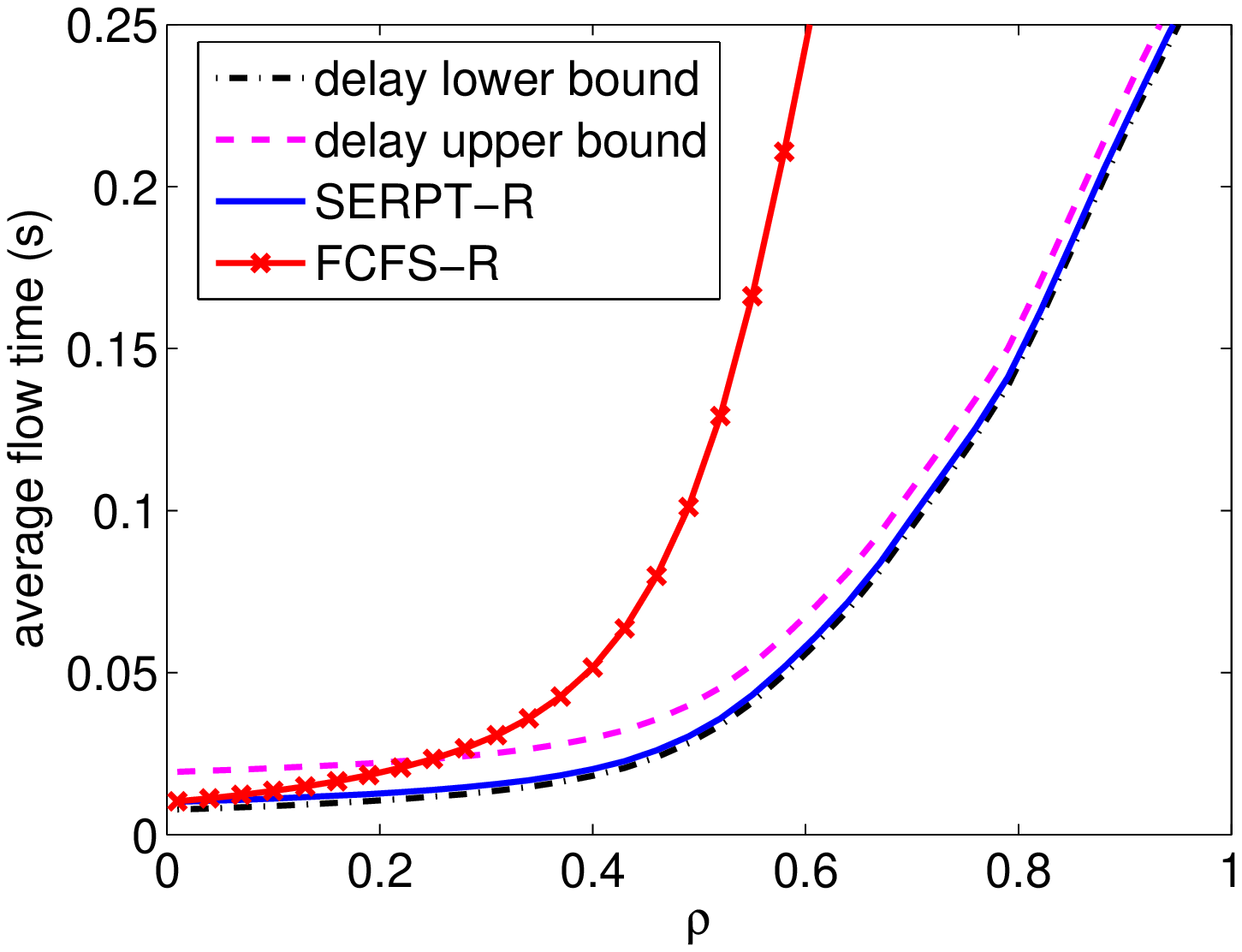}}}}
    \hspace{0.1\textwidth}
{\subfigure[][Preemption is not allowed, $d_{\min}<L=5$]{\resizebox{0.3\textwidth}{!}{\includegraphics{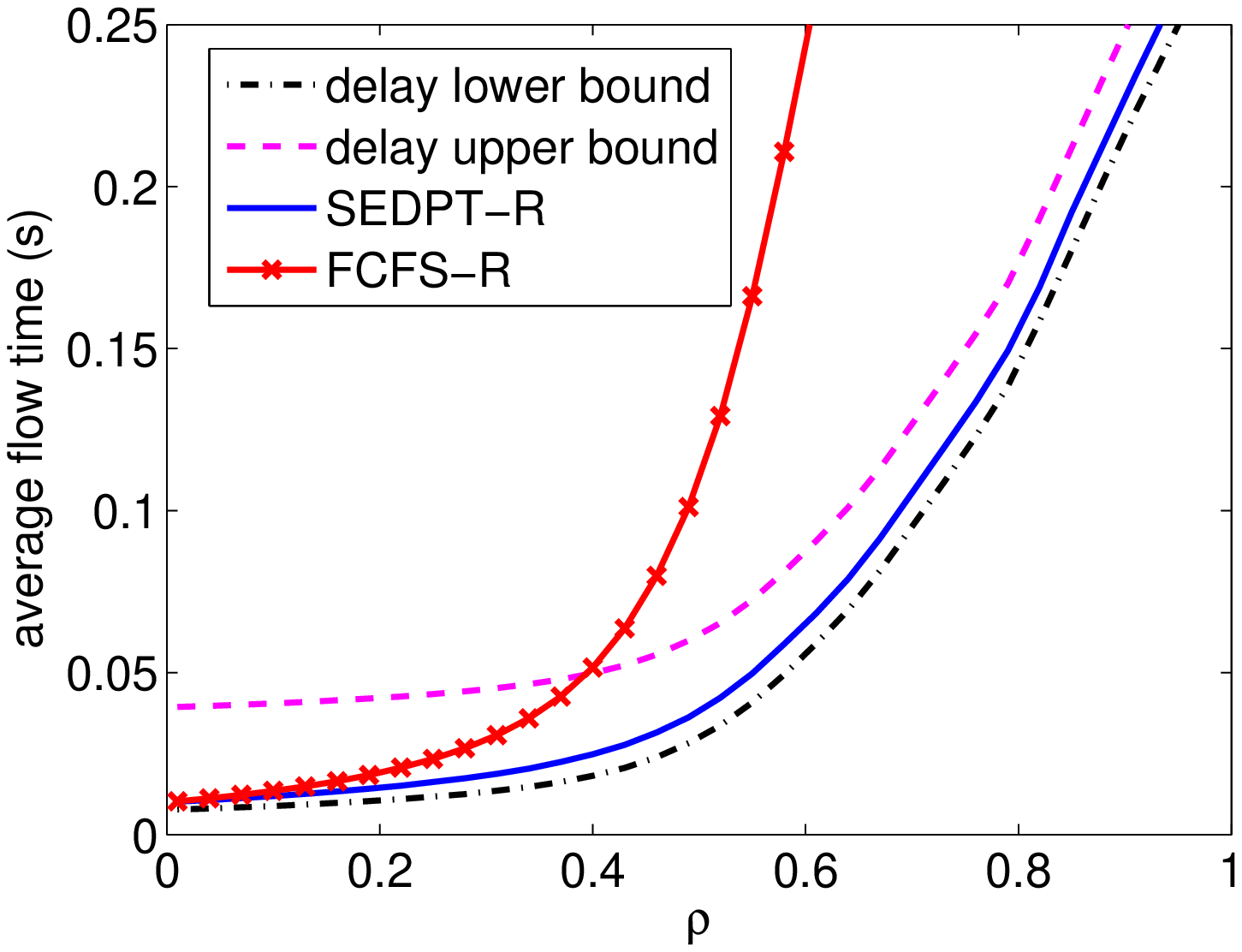}}}}
\vspace{+0.2cm}
\end{center}\vspace{-0.4cm}
\caption{Average flow time $\overline{D}_\pi$ versus traffic intensity $\rho$, where the chunk downloading time is \emph{i.i.d.} exponentially distributed.}\vspace{-0.0cm}
\label{Fig2}
\vspace{-0.4cm}
\end{figure*}

\ifreport
\begin{proof}
See Appendix \ref{app6}.
\end{proof}
\else
\fi

\begin{Theorem}\label{thm5}
Suppose that (i) $d_{\min}\geq L$, (ii) $k_i = 1$ for all $i$,
(iii) preemption is {not} allowed, and (iv) the chunk downloading time is {i.i.d.} NSU. Then, for any given request parameters $N$ and $(a_i,k_i=1,n_i)_{i=1}^N$, non-preemptive SEDPT-R is delay-optimal among all online policies.
\end{Theorem}

\ifreport
\begin{proof}
See Appendix \ref{app7}.
\end{proof}
\else
\fi


A special case of Theorem \ref{thm5} was obtained in Theorem 3 of \cite{shah-Allerton-2013}, where delay-optimality was shown only for high system load such that all $L$ threads are occupied at all time.



\section{Numerical Results}
We present some numerical results to illustrate the delay performance of different scheduling policies and validate the theoretical results. All these results are averaged over 100 random samples for the downloading times of data chunks.

\subsection{Exponential Chunk Downloading Time Distributions}
Consider a system with $N=3000$ request arrivals, among which $p_1=90\%$ of the requested files are stored with a $(n_1,k_1,d_1)=(3,1,3)$ repetition code, and $p_2=10\%$ of the requested files are stored with a $(n_2,k_2,d_2)=(14,10,5)$ Reed-Solomon code. Therefore, $d_{\min}=3$.
The code parameters are drawn at random, \emph{i.i.d.} from these two classes.
The inter-arrival time of the requests is \emph{i.i.d.} distributed as a mixture of exponentials:
\begin{eqnarray}
X\sim\left\{\begin{array}{l l}\text{Exponential}(\text{rate} =0.5\lambda) & \text{with probability}~ 0.99;\\
\text{Exponential}(\text{rate} =50.5\lambda) & \text{with probability}~ 0.01.\end{array}\right.\nonumber
\end{eqnarray}
The average chunk downloading time is $1/\mu$ = 0.02s.
The traffic intensity $\rho$ is determined by
\begin{eqnarray}\label{eq_rho}
\rho = \frac{\left(p_1k_1+p_2k_2 \right)\lambda }{L\mu}.
\end{eqnarray}

Figures \ref{Fig2}(a)-(d) illustrate the numerical results of average flow time $\overline{D}_\pi$ versus traffic intensity $\rho$ for 4 scenarios where the chunk downloading time is \emph{i.i.d.} exponentially distributed. One can observe that SERPT-R and SEDPT-R have shorter average flow times than the First-Come-First-Served policy with Redundant thread assignment (FCFS-R) \cite{ShengboInfocom}. If $L=d_{\min}=3$ and preemption is allowed, by Theorem \ref{thm1}, preemptive SERPT-R is delay-optimal. For the other 3 scenarios, upper and lower bounds of the optimum delay performance are plotted. By comparing with the delay lower bound, we find that the extra delay caused by non-preemption is $0.0114$s which is smaller than $1/\mu= 0.02$s, and the extra delay caused by $d_{\min}<L$ is $0.0034$s which is smaller than $\frac{1}{\mu}\sum_{l=d_{\min}}^{L-1}\frac{1}{l}= 0.0117$s. These results are in accordance with Theorems \ref{thm1}-\ref{thm4}.

\subsection{NLU Chunk Downloading Time  Distributions}
For the NLU distributions, the system setup is the same with that in the previous subsection.
We assume that the chunk downloading time $X$ is \emph{i.i.d.} distributed as the sum of a constant and a value drawn from an exponential distribution:
\begin{eqnarray}
\Pr(X>x)=\left\{\begin{array}{l l} 1, &\text{~if~} x\leq \frac{0.4}{\mu};\\
\exp\left[-\frac{\mu}{0.6} (x-\frac{0.4}{\mu})\right],& \text{~if~} x\geq \frac{0.4}{\mu},\end{array}\right.
\end{eqnarray}
which was proposed in \cite{addShengboOriginal,Liang2013_2} to model the data downloading time in Amazon AWS system.
The traffic intensity $\rho$ is also given by \eqref{eq_rho}.

Figure \ref{fig5} illustrates the average flow time $\overline{D}_\pi$ versus traffic intensity $\rho$ when $L= 3$ and the chunk downloading time is \emph{i.i.d.} NLU. As expected, preemptive SEDPT-WCR has a shorter average delay than non-preemptive SEDPT-NR. In the preemptive case, the delay performance of SEDPT-WCR is much better than those of non-preemptive SEDPT-R and the First-Come-First-Served policy with Work-Conserving Redundant thread assignment (FCFS-WCR). Therefore, preemptive SEDPT-WCR and non-preemptive SEDPT-NR are appropriate for \emph{i.i.d.} NLU downloading time distributions. By comparing with the delay lower bound, we find that the maximum extra delays of preemptive SEDPT-WCR and non-preemptive SEDPT-NR are $0.0229$s and $0.0230$s, respectively. Both of them are smaller than the delay gap in Theorems \ref{thm7} and \ref{thm8}, whose value is $0.0560$s.


\subsection{NSU Chunk Downloading Time Distributions}

For NSU distributions, we consider that all $N=3000$ requested files are stored with a $(n_1,k_1,d_1)=(3,1,3)$ repetition code. The chunk downloading time $X$ is chosen \emph{i.i.d.} as a mixture of exponentials:
\begin{eqnarray}
X\sim\left\{\begin{array}{l l}\text{Exponential}(\text{rate} =0.4\mu) & \text{with probability}~ 0.5;\\
\text{Exponential}(\text{rate} =1.6\mu) & \text{with probability}~ 0.5.\end{array}\right.\nonumber
\end{eqnarray}
Under SEDPT-R, the average time for completing one chunk is $\mathbb{E}\left\{\min_{l=1,\cdots,L} X_l\right\}$, where the $X_l$'s are \emph{i.i.d.} chunk downloading times.
Therefore, the traffic intensity $\rho$ is
\begin{eqnarray}
\rho = \lambda \mathbb{E}\left\{\min_{l=1,\cdots,L} X_l\right\}.
\end{eqnarray}
Figure \ref{fig6} shows the average flow time $\overline{D}_\pi$ versus traffic intensity $\rho$ where $L= 3$, preemption is not allowed, and the chunk downloading time is \emph{i.i.d.} NSU. In this case, SEDPT-R is delay-optimal. We observe that the delay performance of SEDPT-WCR is quite bad and the delay gap between SEDPT-R and SEDPT-WCR is unbounded. This is because SEDPT-WCR has a smaller throughput region than SEDPT-R. Therefore, SEDPT-R is appropriate for \emph{i.i.d.} NSU downloading time distributions.


\section{Conclusions}\label{sec6}
In this paper, we have analytically characterized the delay-optimality of data retrieving in distributed storage systems with multiple storage codes.
\ifreport
Low-latency thread scheduling policies have been designed by combining the advantages of SERPT in the preemptive case (or SEDPT in the non-preemptive case) and redundant thread assignment.
\fi
Under several important settings, we have shown that the proposed policies are either delay-optimal or within a constant gap from the optimum delay performance.

There are several important open problems concerning the analytical characterization of data retrieving delay:

\begin{itemize}
\item What is the optimal policy for other classes of non-exponential service distributions?

\item What is the optimal policy when the service time distributions are heterogeneous across data chunks?

\item What is the optimal policy when latency and downloading cost need to be jointly considered?

\item How to design low-latency policies under delay metrics other than average flow time?
\end{itemize}

\begin{figure}
\centering \includegraphics[width=0.3\textwidth]{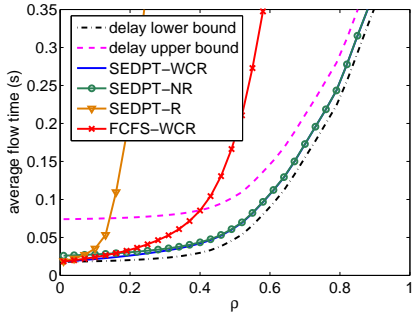} \caption{Average flow time $\overline{D}_\pi$ versus traffic intensity $\rho$, where the chunk downloading time is \emph{i.i.d.} NLU.}
\label{fig5} \vspace{-0.cm}
\end{figure}

\begin{figure}
\centering \includegraphics[width=0.3\textwidth]{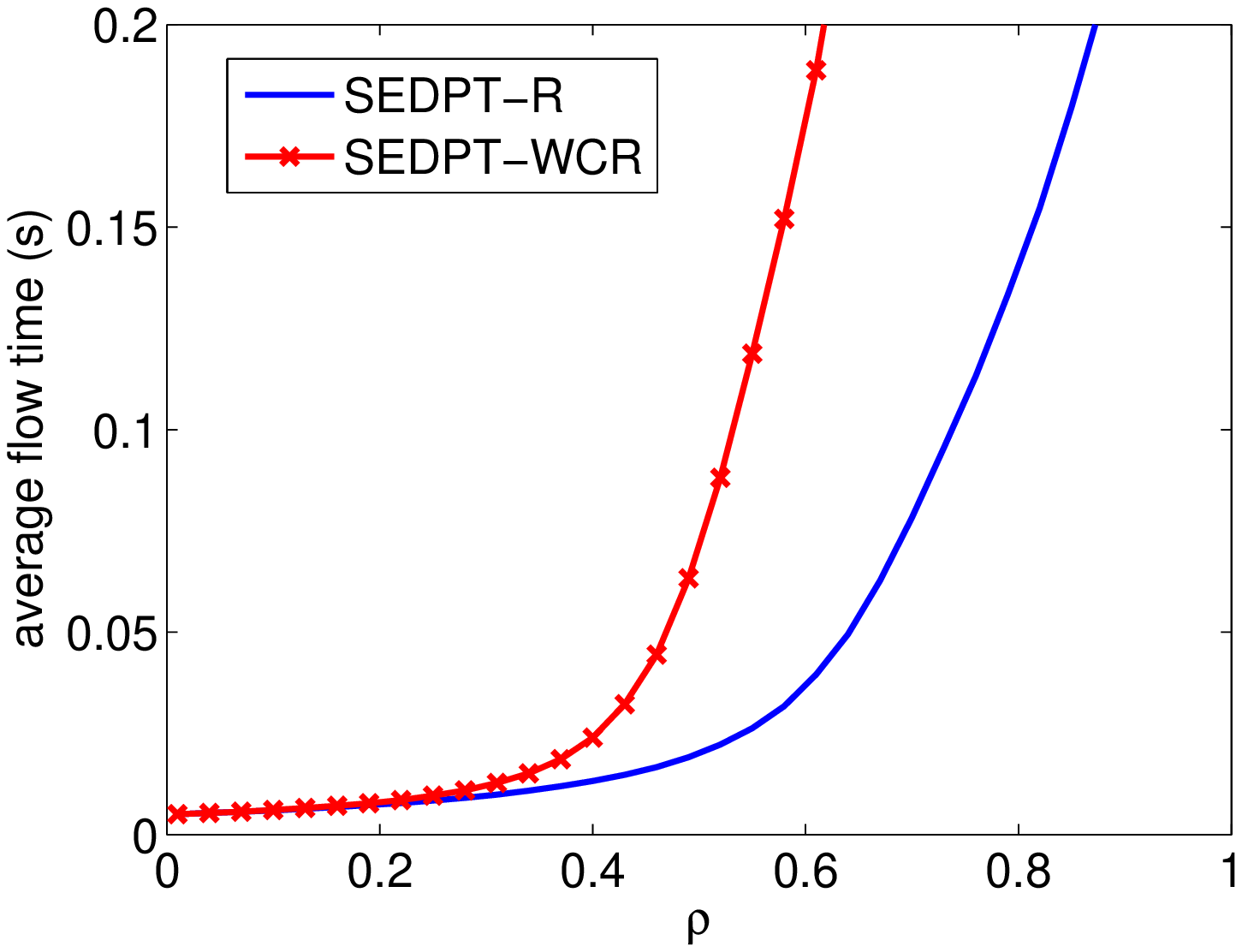} \caption{Average flow time $\overline{D}_\pi$ versus traffic intensity $\rho$, where the chunk downloading time is \emph{i.i.d.} NSU.}
\label{fig6} \vspace{-0.cm}
\end{figure}

\bibliographystyle{ieeetran}
\bibliography{ref}

\ifreport
\appendices
\ignore{
\section{Useful Notations and Definitions}
The decision-maker cannot predict beforehand when and which chunks will be downloaded. However, after some chunks are downloaded, the decision-maker knows about how many additional chunks are needed for completing request.
\begin{definition}
\textbf{Available chunks}: the chunks that have not been downloaded and are not under service.
\end{definition}

\begin{definition}
\textbf{Remaining chunks}: from the set of available chunks, any subset of chunks that are sufficient for retrieving each file.
\end{definition}

Let $u_{i}(t)$ denote the number of available chunks of file $i$ at time $t$, and $r_i(t)$ denote the number of remaining chunks of file $i$ at time $t$. The parameters $(u_i(t),r_i(t))$ will be used to make scheduling decisions. Suppose that $g_i(t)$ chunks of file $i$ have been downloaded by time $t$. If $g_i(t)< k_i$, then require $i$ is not completed at time $t$, such that $r_i(t)=k_i-g_i(t)$ and $u_i(t)=n_i-g_i(t)$.
Therefore, any unfinished request $i$ satisfies
\begin{eqnarray}\label{eq_u}
u_i(t)-r_{i}(t) = n_i-k_i = d_i-1.
\end{eqnarray}
If $g_i(t)= k_i$, then request $i$ is completed before time $t$ such that $r_i(t)=u_i(t)=0$.
We further define
\begin{eqnarray}\label{eq_def4}
r(t)=\sum_{i:a_i\leq t} r_i(t)
\end{eqnarray}
as the total number of remaining chunks of all unfinished requests at time $t$, and
\begin{eqnarray}\label{eq_def3}
u(t)=\sum_{i:a_i\leq t} u_i(t)
\end{eqnarray}
as the total number of available chunks of all unfinished requests at time $t$.
}
\section{Proof of Theorem~\ref{thm1}} \label{app0}
{First, consider an arbitrarily given sample path of chunk departures $(t_1,t_2,\ldots)$. According to the conditions of Theorem \ref{thm1}, the request parameters $N$ and $(a_i,k_i,n_i)_{i=1}^N$ are fixed. Then, the request completion times $(c_{1,\pi},c_{2,\pi},\ldots)$ of a policy $\pi$ are determined by which request each departed chunk belongs.
Let $D_\pi (t_1,t_2,\ldots)=\frac{1}{N}\sum_{i=1}^N (c_{i,\pi}-a_i)$ denote the sample-path average delay of policy $\pi$ for given request parameters $N$, $(a_i,k_i,n_i)_{i=1}^N$ and chunk departures $(t_1,t_2,\ldots)$. According to the SRPT discipline \cite{Schrage68,Smith78},  $D_\pi (t_1,t_2,\ldots)$ is minimized if each downloaded chunk belongs to the request with the fewest remaining chunks. This is satisfied by preemptive SERPT-R under the conditions of Theorem \ref{thm1}, because all $L$ threads are assigned to the request with the fewest remaining chunks. Therefore, for any given chunk departures $(t_1,t_2,\ldots)$, preemptive SERPT-R minimizes $D_\pi (t_1,t_2,\ldots)$, i.e.,
\begin{eqnarray}\label{eq_proof1_1}
D_{\text{SERPT-R}} (t_1,t_2,\ldots) = \min_{\pi} D_\pi (t_1,t_2,\ldots).
\end{eqnarray}

Let $F_\pi (t_1,t_2,\ldots)$ denote the cumulative distribution function of the chunk departure process $(t_1,t_2,\ldots)$ under policy $\pi$. Then, the average delay of policy $\pi$ can be expressed as
\begin{eqnarray}
\overline{D}_\pi =
\int D_\pi (t_1,t_2,\ldots) d F_\pi (t_1,t_2,\ldots).
\end{eqnarray}
According to Lemma \ref{lem0-1}, any two work-conserving policies $\pi_1$ and $\pi_2$ satisfy
\begin{eqnarray}\label{eq_proof1_2}
F_{\pi_1} (t_1,t_2,\ldots) = F_{\pi_2} (t_1,t_2,\ldots), ~\forall~(t_1,t_2,\ldots).
\end{eqnarray}
Using \eqref{eq_proof1_1}-\eqref{eq_proof1_2} and the fact that preemptive SERPT-R is a work-conserving policy, we can obtain for any work-conserving policy $\pi$ that
\begin{eqnarray}
&&\!\!\!\!\!\!\!\!\!~~~\overline{D}_\pi \nonumber\\
&&\!\!\!\!\!\!\!\!\!= \int D_{\pi} (t_1,t_2,\ldots) d F_\pi (t_1,t_2,\ldots) \nonumber\\
&&\!\!\!\!\!\!\!\!\!\geq \int D_{\text{SERPT-R}} (t_1,t_2,\ldots) d F_\pi (t_1,t_2,\ldots)\nonumber\\
&&\!\!\!\!\!\!\!\!\!=  \int D_{\text{SERPT-R}} (t_1,t_2,\ldots) d F_\text{SERPT-R} (t_1,t_2,\ldots)\nonumber\\
&&\!\!\!\!\!\!\!\!\!= \overline{D}_{\text{SERPT-R}}.
\end{eqnarray}
Hence, preemptive SERPT-R is delay-optimal among the class of work-conserving policies.
Finally, when preemption is allowed, a delay-optimal policy must be work-conserving. Hence, Theorem \ref{thm1} follows.}

\section{Proof of Theorem~\ref{thm3}} \label{app1}
The case of $L\leq d_{\min}$ was studied in Theorem \ref{thm1} and we only need to consider the case of $L> d_{\min}$. For notational simplicity, we use policy $P$ to denote preemptive SERPT-R with $L> d_{\min}$, and policy $Q$ to denote preemptive SERPT-R under the conditions of Theorem \ref{thm1} where $L\leq d_{\min}$ holds. In particular, policy $P$ is under the request parameters $N$ and $(a_i, k_i,n_i)_{i=1}^N$ such that there exists an integer $j$ $(1\leq j\leq N)$ satisfying $L> n_j-k_j +1$, and policy $Q$ has some ``virtual'' chunks such that it is under the request parameters $N$ and $(a_i, k_i,n'_i)_{i=1}^N$ satisfying $L\leq n'_i-k_i +1$ for all $i=1,2,\ldots,N$.

When $L> d_{\min}$, the optimal policy of Theorem \ref{thm3} can be an non-work-conserving policy under the conditions of Theorem \ref{thm1}, because there can be less than $L$ available chunks to download. By Theorem \ref{thm1}, policy $Q$ provides a lower bound of $\overline{D}_{\text{opt}}$. On the other hand, policy $P$ provides an upper bound of $\overline{D}_{\text{opt}}$. The remaining task is to evaluate the delay gap between policy $P$ and policy $Q$ when $L> d_{\min}$.

First, we construct the chunk departure sample paths $(t_1,t_2,\ldots)$ of policy $P$ and policy $Q$. Let $(t_1^l, t_2^l,\ldots)$ denote the chunk departure time sequences of thread $l$, such that the inter-departure time $\tau_j^l = t_{j+1}^l-t_j^l$ is i.i.d. exponentially distributed with rate $\mu$. Under policy $P$, the chunk departure time sequences $(t_1,t_2,\ldots)$ is obtained by taking the union $\cup_{l=1}^L(t_1^l, t_2^l,\ldots)$ and deleting the chunk departures during the idle periods of each thread $l$ under policy $P$. (Under policy $P$, the idle periods are different across the threads.) Since the chunk service time is memoryless, deleting some chunk departures will not affect the service time distribution of other chunks.
Under policy $Q$, the chunk departure time sequences $(t_1,t_2,\ldots)$ is obtained by taking the union $\cup_{l=1}^L(t_1^l, t_2^l,\ldots)$, and deleting the chunk departures when all $L$ threads are idle under policy $Q$. (Under policy $Q$, all $L$ threads are active or idle at the same time.) 
By this, we obtain two chunk departure sample paths of policy $P$ and policy $Q$ with the same probability to occur.

In the sequel, we will show that \emph{for any time $t$ and chunk departure sample paths of policy $P$ and policy $Q$ constructed above, policy $P$ needs to download $L-d_{\min}$ or fewer additional chunks after time $t$, so as to accomplish the same number of requests that are completed under policy $Q$ during $(0,t]$}.

\begin{definition}\cite{Smith78}
The state of the system is specified by an infinite vector $\vec{\alpha}=(\alpha_1,\alpha_2,\ldots)$ with non-negative, non-increasing components. At any time, the coordinates of $\vec{\alpha}$ are interpreted as follows: $\alpha_1$ is the maximum number of remaining chunks among all requests, $\alpha_2$ is the next greatest number of remaining chunks among all requests, and so on, with duplications being explicitly repeated. Suppose that there are $l$ unfinished requests in the system, then
\begin{eqnarray}
\alpha_1\geq\alpha_2\geq\ldots\geq\alpha_l>0 = \alpha_{l+1} =\alpha_{l+2}=\ldots.
\end{eqnarray}
\end{definition}

The key step for proving Theorem \ref{thm3} is to establish the following result:
\begin{lemma}\label{lemB0}
Let $\{\vec{\alpha}(t),t\geq0\}$ be the state process of policy $P$ and $\{\vec{\beta}(t),t\geq0\}$ be the state process of policy $Q$. If $L> d_{\min}$ and $\vec{\alpha}(0) = \vec{\beta}(0)$, then for the chunk departure sample paths of policy $P$ and policy $Q$ described above, we have
\begin{eqnarray}
\sum_{i=j}^\infty \alpha_i(t) \leq \sum_{i=j}^\infty \beta_i(t) + L- d_{\min}
\end{eqnarray}
for all $t\geq0$ and j = $1,2,\ldots$
\end{lemma}

In order to prove this result, we first establish the following lemmas:

\begin{lemma}\label{lemB1}
Suppose that, under policy $P$, the system state at time $t$ is $\vec{\alpha}$ and at time $t+\Delta t$ is $\vec{\alpha}'$. Further, suppose that, under policy $Q$, the system state at time $t$ is $\vec{\beta}$ and at time $t+\Delta t$ is $\vec{\beta}'$. If (i) $L> d_{\min}$, (ii) no arrivals occur during the interval $(t, t+\Delta t]$ and (iii)
\begin{eqnarray}\label{eq_41}
\sum_{i=j}^\infty \alpha_i \leq \sum_{i=j}^\infty \beta_i + L- d_{\min}, ~\forall~j=1,2,\ldots
\end{eqnarray}
Then, for the chunk departure sample paths of policy $P$ and policy $Q$ described above, we have
\begin{eqnarray}\label{eq_40}
\sum_{i=j}^\infty \alpha_i' \leq \sum_{i=j}^\infty \beta_i' + L- d_{\min}, ~\forall~j=1,2,\ldots
\end{eqnarray}
\end{lemma}

\begin{proof}
If $\sum_{i=j}^\infty \alpha_i'\leq L - d_{\min}$, then \eqref{eq_40} follows naturally.

If $\sum_{i=j}^\infty \alpha_i'\geq L- d_{\min}+1$, the unfinished requests have at least $L- d_{\min}+1$ remaining chunks to download at time $t+\Delta t$. Equation \eqref{eq_d_i} tells us that each unfinished request $i$ has $n_i-k_i = d_{i}-1$ redundant chunks. Therefore, the system must have at least a total number of $\sum_{i=j}^\infty \alpha_i'+ d_{\min}-1 \geq L$ available chunks at time $t+\Delta t$, and all $L$ threads are active under policy $P$ at time $t+\Delta t$.

Next, since there is no request arrivals during the interval $(t, t+\Delta t]$, all $L$ threads must be kept active during $(t, t+\Delta t]$ under policy $P$. Suppose that $b$ chunks are downloaded under policy $P$ during $(t, t+\Delta t]$. Then, in the two chunk departure sample paths constructed above, no more than $b$ chunks are downloaded under policy $Q$ during $(t, t+\Delta t]$, because the threads can be idle.

Further, suppose that one chunk being served at time $t+\Delta t$ under policy $P$ is associated to an $\alpha_m'$ satisfying $\alpha_m'>\alpha_j'$. Then, according to the description of policy $P$ (preemptive SERPT-R), all the available chunks of the requests with $\alpha_j'$ or fewer remaining chunks must be also under service at time $t+\Delta t$. We have just shown that the requests with $\alpha_j'$ or fewer remaining chunks have a total number of at least $L$ available chunks. Thus, the total number of chunks under service at time $t+\Delta t$ is no less than $L + 1$, which is impossible. Therefore, any request under service at time $t+\Delta t$ must associate to an $\alpha_m'$ satisfying $\alpha_m'\leq \alpha_j'$. Since no arrivals occur during the interval $(t, t+\Delta t]$, each downloaded chunk of policy $P$ during $(t, t+\Delta t]$ must belong to some request associated to an $\alpha_m'$ satisfying $\alpha_m'\leq\alpha_j'$.

Using these facts, we can obtain
$\sum_{i=j}^\infty \alpha_i' = \sum_{i=j}^\infty \alpha_i - b\leq \sum_{i=j}^\infty \beta_i + L- d_{\min} - b \leq \sum_{i=j}^\infty \beta_i' + L- d_{\min}$, where the equality is due to the fact that each downloaded chunk of policy $P$ must belong to some request associated to an $\alpha_m'$ satisfying $\alpha_m'\leq\alpha_j'$, the first inequality is due to \eqref{eq_41}, and the second inequality is due to the fact that no more than $b$ chunks are downloaded under  policy $Q$ during $(t, t+\Delta t]$.
\end{proof}

\begin{lemma}\label{lemB2}
Suppose that, under policy $P$, $\vec{\alpha}'$ is obtained by adding a request with $b$ remaining chunks to the system whose state is $\vec{\alpha}$. Further, suppose that, under policy $Q$, $\vec{\beta}'$ is obtained by adding a request with $b$ remaining chunks to the system whose state is $\vec{\beta}$.
If
\begin{eqnarray}
\sum_{i=j}^\infty \alpha_i \leq \sum_{i=j}^\infty \beta_i + L- d_{\min}, ~\forall~j=1,2,\ldots,
\end{eqnarray}
then
\begin{eqnarray}
\sum_{i=j}^\infty \alpha_i' \leq \sum_{i=j}^\infty \beta_i' + L- d_{\min}, ~\forall~j=1,2,\ldots
\end{eqnarray}
\end{lemma}

\begin{proof}
The proof is similar to Lemma 3 in \cite{Smith78}.
Without loss of generalization, we suppose that $b$ is the $l$th coordinate of $\vec{\alpha}'$ and the $m$th coordinate of $\vec{\beta}'$. We consider the following four cases:

\textbf{Case 1}: $l<j, m<j$. We can obtain $\sum_{i=j}^\infty \alpha_i' = \sum_{i=j-1}^\infty \alpha_i \leq \sum_{i=j-1}^\infty \beta_i + L- d_{\min} = \sum_{i=j}^\infty \beta_i' + L- d_{\min}$.

\textbf{Case 2}: $l<j, m\geq j$. We have $\sum_{i=j}^\infty \alpha_i' = \sum_{i=j-1}^\infty \alpha_i \leq b + \sum_{i=j}^\infty \alpha_i \leq b + \sum_{i=j}^\infty \beta_i + L- d_{\min} = \sum_{i=j}^\infty \beta_i' + L- d_{\min}$.

\textbf{Case 3}: $l\geq j, m< j$. We have $\sum_{i=j}^\infty \alpha_i' = b + \sum_{i=j}^\infty \alpha_i \leq \sum_{i=j-1}^\infty \alpha_i \leq \sum_{i=j-1}^\infty \beta_i + L- d_{\min} = \sum_{i=j}^\infty \beta_i' + L- d_{\min}$.

\textbf{Case 4}: $l\geq j, m\geq j$. We have $\sum_{i=j}^\infty \alpha_i' = b + \sum_{i=j}^\infty \alpha_i \leq b + \sum_{i=j}^\infty \beta_i + L- d_{\min} = \sum_{i=j}^\infty \beta_i' + L- d_{\min}$.
\end{proof}

Using the initial state $\vec{\alpha}(0)=\vec{\beta}(0)$, Lemmas \ref{lemB1} and \ref{lemB2}, it is straightforward to prove Lemma \ref{lemB0}.
After Lemma \ref{lemB0} is established, we are ready to prove Theorem~\ref{thm3}.

\begin{proof}[Proof of Theorem~\ref{thm3}]
As explained above, we only need to evaluate the delay gap between policy $P$ and policy $Q$ when $L> d_{\min}$. Let the evolution of the system state under some queueing discipline be on a space $(\Omega,\mathcal{F},P)$. We assume that the request arrival process $\{a_i,k_i,n_i\}_{i=1}^N$ is fixed for all $\omega\in\Omega$. Let $\{\vec{\alpha}(t),t\geq0\}$ be the state process of policy $P$ and $\{\vec{\beta}(t),t\geq0\}$ be the state process of policy $Q$. Then, we have $\vec{\alpha}(0)=\vec{\beta}(0)=0$.

Suppose that under policy $Q$, there are $y$ request arrivals and $z$ request departures during $(0, t]$. Then, there are $y-z$ requests in the system at time $t$ such that $\sum_{i=y-z+1}^\infty \beta_i(t) = 0$. According to Lemma \ref{lemB0}, we have $\sum_{i=y-z+1}^\infty \alpha_i(t) \leq L- d_{\min}$. Hence, under policy $P$, the system still needs to download $L-d_{\min}$ or fewer chunks after time $t$, in order to complete $z$ requests as in policy $Q$. 
Suppose that exactly $L-d_{\min}$ chunks are needed to complete $z$ requests. At time $t$, at least $L-1$ threads are assigned to serve the requests associated to the $L-d_{\min}$ chunks that are most likely to result in request departures. After one of these chunks is downloaded, at least $L-2$ threads are assigned to serve the requests associated to the $L-d_{\min}-1$ chunks that are most likely to result in request departures. This procedure goes on, until $L-d_{\min}$ chunks are downloaded. Because the chunk download time of each thread is i.i.d. exponentially distributed with mean $1/\mu$, the average time for downloading these $L-d_{\min}$ chunks under policy $P$ is upper bounded by
\begin{eqnarray}\label{eq_delay111}
\sum_{l=d_{\min}}^{L-1}\frac{1}{l\mu},
\end{eqnarray}
where $\frac{1}{l\mu}$ is the average time for downloading one chunk when $l$ threads are active. If less than $L-d_{\min}$ chunks are needed to complete $z$ requests, the average  downloading time will be even shorter. Hence, the delay gap between policy $P$ and policy $Q$ is no more than the term in \eqref{eq_delay111}, and \eqref{eq_4-1-1} follows.
\end{proof}

\section{Proof of Theorem~\ref{thm2}} \label{app2}
First, the optimal policy under the conditions of Theorem \ref{thm2} is feasible even if preemption is allowed. Hence, by Theorem \ref{thm1}, preemptive SERPT-R provides a lower bound of $\overline{D}_{\text{opt}}$, i.e., the optimal delay of the policies satisfying the conditions of Theorem \ref{thm2}. On the other hand, non-preemptive SEDPT-R provides an upper bound of $\overline{D}_{\text{opt}}$.
The remaining task is to evaluate the delay gap between preemptive SERPT-R and non-preemptive SEDPT-R.

For notational simplicity, we use policy $P$ to denote preemptive SERPT-R, and policy $NP$ to denote non-preemptive SEDPT-R. We will show that \emph{for any time $t$ and any given sample path of chunk departures $(t_1,t_2,\ldots)$, policy $NP$ needs to download $L$ or fewer additional chunks after time $t$, so as to accomplish the same number of requests that are completed under policy $P$ during $(0,t]$}.

\begin{definition}\cite{Smith78}
The system state of \emph{preemptive SERPT-R (policy $P$)} is specified by an infinite vector $\vec{\beta}=(\beta_1,\beta_2,\ldots)$ with non-negative, non-increasing components. At any time, the coordinates of $\vec{\beta}$ are interpreted as follows: $\beta_1$ is the maximum number of remaining chunks among all requests, $\beta_2$ is the next greatest number of remaining chunks among all requests, and so on, with duplications being explicitly repeated. Suppose that there are $l$ unfinished requests in the system, then
\begin{eqnarray}
\beta_1\geq\beta_2\geq\ldots\geq\beta_l>0 = \beta_{l+1} =\beta_{l+2}=\ldots.
\end{eqnarray}
\end{definition}

\begin{definition}
The system state of \emph{non-preemptive SEDPT-R (policy $NP$)} is specified by a pair of vectors $\{\vec{\alpha},\vec{\delta}\}$, where $\vec{\alpha}=(\alpha_1,\alpha_2,\ldots)$ and $\vec{\delta}=(\delta_1,\delta_2,\ldots)$ are two infinite vectors with non-negative components. At any time, the coordinates of $\vec{\alpha}$ and $\vec{\delta}$ are interpreted as follows: $\alpha_i$ is the number of chunks to be downloaded for completing the request associated to the $i$th coordinate, and $\delta_i$ is the number of threads assigned to serve the request associated to the $i$th coordinate such that $\sum_{i=1}^\infty \delta_i\leq L$. Suppose that there are $l$ unfinished requests in the system, then the coordinates of $\vec{\alpha}$ and $\vec{\delta}$ are sorted such that
\begin{eqnarray}
&&\alpha_1-\delta_1\geq\alpha_2-\delta_2\geq\ldots\geq \alpha_{l}-\delta_{l},\\
&&\alpha_{l+1}-\delta_{l+1} = \alpha_{l+2}-\delta_{l+2}=\ldots = 0,\\
&&\alpha_i \left\{\begin{array}{l l} >0, &\text{if}~ i\leq l;\\
  =0, &\text{if}~ i\geq l+1,\end{array}\right.\\
&&\delta_i \left\{\begin{array}{l l} \geq0, &\text{if}~ i\leq l;\\
  =0, &\text{if}~ i\geq l+1.\end{array}\right.
\end{eqnarray}
Note that there exists an integer $i$ ($0\leq i\leq l$) such that $\alpha_1-\delta_1\geq\ldots\geq \alpha_{i}-\delta_{i}> 0\geq\alpha_{i+1}-\delta_{i+1}\geq\ldots\geq \alpha_{l}-\delta_{l}$.
\end{definition}

The key step for proving Theorem \ref{thm2} is to establish the following result:
\begin{lemma}\label{lem_non_prmp0}
Let $\{\vec{\alpha}(t),\vec{\delta}(t),t\geq0\}$ be the state process of policy $NP$ and $\{\vec{\beta}(t),t\geq0\}$ be the state process of policy $P$. If $\vec{\alpha}(0)=\vec{\delta}(0) = \vec{\beta}(0) =0$, then for any given sample path of chunk departures $(t_1,t_2,\ldots)$, we have
\begin{eqnarray}
\sum_{i=j}^\infty [\alpha_i(t)- \delta_i(t)]\leq \sum_{i=j}^\infty \beta_i(t)
\end{eqnarray}
for all $t\geq0$ and j = $1,2,\ldots$
\end{lemma}

In order to prove this result, we first establish the following lemmas:

\begin{lemma}\label{lem_non_prmp1}
Suppose that, under policy $NP$, $\{\vec{\alpha}',\vec{\delta}'\}$ is obtained by completing a chunk at one of the $L$ threads in the system whose state is $\{\vec{\alpha},\vec{\delta}\}$. Further, suppose that, under policy $P$, $\vec{\beta}'$ is obtained by completing a chunk at one of the $L$ threads in the system whose state is $\vec{\beta}$.
If
\begin{eqnarray}\label{eq_non_prmp_41}
\sum_{i=j}^\infty [\alpha_i - \delta_i]\leq \sum_{i=j}^\infty \beta_i, ~\forall~j=1,2,\ldots,
\end{eqnarray}
then
\begin{eqnarray}\label{eq_non_prmp_40}
\sum_{i=j}^\infty [\alpha'_i - \delta'_i]\leq \sum_{i=j}^\infty \beta'_i, ~\forall~j=1,2,\ldots
\end{eqnarray}
\end{lemma}

\begin{proof}
Suppose that, under policy $NP$, there are $l$ unfinished requests at state $\{\vec{\alpha},\vec{\delta}\}$. If $\sum_{i=j}^\infty [\alpha'_i - \delta'_i]\leq0$, then the inequality \eqref{eq_non_prmp_40} follows naturally. In the following, we will consider the scenario of $\sum_{i=j}^\infty [\alpha'_i - \delta'_i]>0$ in two cases.

\textbf{Case 1}: Under policy $NP$, the chunk departure does not lead to a request completion. In this case, the thread that has just completed a chunk will be reassigned to serve the request associated to the $l$th coordinate such that $\alpha'_l- \delta'_l = \alpha_l- \delta_l -1$. Meanwhile, we have $\alpha'_i- \delta'_i = \alpha_i- \delta_i$ for all $i=1,2,\ldots,l-1$, and $\alpha'_i- \delta'_i = 0$ for all $i=l+1,l+2,\ldots$ Since $\sum_{i=j}^\infty [\alpha'_i - \delta'_i]>0$, we have $j\leq l$. Therefore, $\sum_{i=j}^\infty [\alpha'_i - \delta'_i]= \sum_{i=j}^\infty [\alpha_i - \delta_i] -1 \leq \sum_{i=j}^\infty \beta_i -1 \leq \sum_{i=j}^\infty \beta'_i$.

\textbf{Case 2}: Under policy $NP$, the chunk departure results in a request departure. Suppose that the departed request is associated to the $m$th coordinate at state $\{\vec{\alpha},\vec{\delta}\}$ ($m\leq l$). After the request departure, the threads that was previous serving the request associated to the $m$th coordinate will be reassigned to serve the request associated to the $l-1$th coordinate at state $\{\vec{\alpha}',\vec{\delta}'\}$.

If $j\geq m$, then we have $\sum_{i=j}^\infty [\alpha'_i - \delta'_i] \leq \sum_{i=j+1}^\infty [\alpha_i - \delta_i] -1 \leq \sum_{i=j+1}^\infty \beta_i -1\leq \sum_{i=j}^\infty \beta_i -1 \leq \sum_{i=j}^\infty \beta'_i$.

If $j< m$, then we have $\sum_{i=j}^\infty \alpha'_i  = \sum_{i=j}^\infty \alpha_i -1$ and $\sum_{i=j}^\infty \delta'_i  = \sum_{i=j}^\infty \delta_i$. Hence, $\sum_{i=j}^\infty [\alpha'_i - \delta'_i] = \sum_{i=j}^\infty [\alpha_i - \delta_i] -1 \leq \sum_{i=j}^\infty \beta_i -1\leq \sum_{i=j}^\infty \beta'_i$.
\end{proof}

\begin{lemma}\label{lem_non_prmp2}
Suppose that, under policy $NP$, $\{\vec{\alpha}',\vec{\delta}'\}$ is obtained by adding a request with $b$ remaining chunks to the system whose state is $\{\vec{\alpha},\vec{\delta}\}$. Further, suppose that, under policy $P$, $\vec{\beta}'$ is obtained by adding a request with $b$ remaining chunks to the system whose state is $\vec{\beta}$.
If
\begin{eqnarray}
\sum_{i=j}^\infty [\alpha_i - \delta_i]\leq \sum_{i=j}^\infty \beta_i, ~\forall~j=1,2,\ldots,
\end{eqnarray}
then
\begin{eqnarray}
\sum_{i=j}^\infty [\alpha'_i - \delta'_i]\leq \sum_{i=j}^\infty \beta'_i, ~\forall~j=1,2,\ldots
\end{eqnarray}
\end{lemma}

\begin{proof}
The proof is similar to Lemma 3 in \cite{Smith78}.
Without loss of generalization, we suppose that $b$ is the $l$th coordinate of $\{\vec{\alpha}',\vec{\delta}'\}$ and the $m$th coordinate of $\vec{\beta}'$. We consider the following four cases:

\textbf{Case 1}: $l<j, m<j$. We can obtain $\sum_{i=j}^\infty [\alpha'_i - \delta'_i] = \sum_{i=j-1}^\infty [\alpha_i - \delta_i] \leq \sum_{i=j-1}^\infty \beta_i= \sum_{i=j}^\infty \beta_i'$.

\textbf{Case 2}: $l<j, m\geq j$. We have $\sum_{i=j}^\infty [\alpha'_i - \delta'_i] = \sum_{i=j-1}^\infty [\alpha_i - \delta_i] \leq b + \sum_{i=j}^\infty [\alpha_i - \delta_i] \leq b + \sum_{i=j}^\infty \beta_i = \sum_{i=j}^\infty \beta_i'$.

\textbf{Case 3}: $l\geq j, m< j$. We have $\sum_{i=j}^\infty [\alpha'_i - \delta'_i] = b + \sum_{i=j}^\infty [\alpha_i - \delta_i] \leq \sum_{i=j-1}^\infty [\alpha_i - \delta_i] \leq \sum_{i=j-1}^\infty \beta_i = \sum_{i=j}^\infty \beta_i'$.

\textbf{Case 4}: $l\geq j, m\geq j$. We have $\sum_{i=j}^\infty [\alpha'_i - \delta'_i] = b + \sum_{i=j}^\infty [\alpha_i - \delta_i] \leq b + \sum_{i=j}^\infty \beta_i = \sum_{i=j}^\infty \beta_i'$.
\end{proof}

Using Lemma \ref{lem_non_prmp1}, Lemma \ref{lem_non_prmp2}, and the initial state $\vec{\alpha}(0)=\vec{\delta}(0) = \vec{\beta}(0) =0$ at time $t=0$, Lemma \ref{lem_non_prmp0} follows immediately.
After Lemma \ref{lem_non_prmp0} is established, we are ready to prove Theorem~\ref{thm2}.

\begin{proof}[Proof of Theorem~\ref{thm2}]
As explained above, we only need to evaluate the delay gap between policy $NP$ and policy $P$. Let the evolution of the system state under some queueing discipline be on a space $(\Omega,\mathcal{F},P)$. We assume that the request arrival process $\{a_i,k_i,n_i\}_{i=1}^N$ is fixed for all $\omega\in\Omega$. Let $\{\vec{\alpha}(t), \vec{\delta}(t),t\geq0\}$ be the state process of policy $NP$ and $\{\vec{\beta}(t),t\geq0\}$ be the state process of policy $P$. Then, we have $\vec{\alpha}(0)=\vec{\beta}(0)=\vec{\delta}(0)=0$.

Suppose that under policy $P$, there are $y$ request arrivals and $z$ request departures during $(0, t]$. Then, there are only $y-z$ unfinished requests in the system at time $t$ such that $\sum_{i=y-z+1}^\infty \beta_i(t) = 0$. According to Lemma \ref{lem_non_prmp0}, we have $\sum_{i=y-z+1}^\infty \alpha_i(t) \leq \sum_{i=y-z+1}^\infty \delta_i(t)$. Hence, under policy $NP$, the system still needs to download $\sum_{i=y-z+1}^\infty \delta_i(t)$ or fewer chunks associated to $\alpha_{y-z+1}(t), \alpha_{y-z+2}(t),\ldots$ after time $t$, in order to complete $z$ requests as in policy $P$. 

Suppose that exactly $\sum_{i=y-z+1}^\infty \delta_i(t)$ chunks are needed to complete $z$ requests. At time $t$, there are $\sum_{i=1}^{y-z} \delta_i(t)$ threads that are assigned to other requests. In order to accomplish $z$ requests, the system still needs to download $\sum_{i=y-z+1}^\infty \delta_i(t)$ chunks associated to $\alpha_{y-z+1}(t), \alpha_{y-z+2}(t),\ldots$, during which time at most
$\sum_{i=1}^{y-z} \delta_i(t)$ chunks associated to $\alpha_{1}(t), \alpha_{2}(t),\ldots, \alpha_{y-z}(t)$ will be downloaded. This is because each thread that is serving a request associated to $\alpha_{1}(t), \alpha_{2}(t),\ldots, \alpha_{y-z}(t)$ at time $t$ will be reassigned to serve a request associated to $\alpha_{y-z+1}(t), \alpha_{y-z+2}(t),\ldots$ after completing the current chunk.
Since $\sum_{i=1}^{\infty} \delta_i(t) \leq L$, the system needs to download at most $L$ extra chunks to complete $z$ requests, regardless of how many of these extra chunks belong to each request. Because the chunk download time of each thread is i.i.d. exponentially distributed with mean $1/\mu$, the average time for the system to use $L$ threads to download $L$ chunks is $1/\mu$. If less than $\sum_{i=y-z+1}^\infty \delta_i(t)$ chunks are needed to complete $z$ requests, the average downloading time will be even shorter. Hence, the delay gap between policy $NP$ and policy $P$ is no more than $1/\mu$, and \eqref{eq_2-1} follows.
\end{proof}

\section{Proof of Theorem~\ref{thm4}}\label{app3-1}
The delay lower bound of $\overline{D}_{\text{opt}}$ is trivial. For the upper bound of $\overline{D}_{\text{opt}}$, we need to combine the proof techniques of Theorem \ref{thm2} and Theorem \ref{thm3} to qualify the delay gap between preemptive SERPT-R and non-preemptive SEDPT-R under the conditions of Theorem~\ref{thm4}. By this, we can show that the delay gap is upper bounded by the average time for downloading $L$ extra chunks due to non-preemption and $L-d_{\min}$ extra chunks due to low storage redundancy. Note that we only need to evaluate the extra delay caused by non-preemption during the time intervals when all $L$ threads are active. This is because when the number of active threads is less than $L$, all the available chunks of the unfinished requests are under service at the same time, and thus non-preemption causes no additional delay beside the extra delay caused by low storage redundancy. By this, Theorem~\ref{thm4} follows.
\section{Proof of Lemma \ref{lem2}}\label{app3}
We first compare the chunk departure time instants among the class of work-conserving policies.

Consider the departure time of the first chunk $t_1$. Because $a_1=s_1=0$ and all $L$ threads are active for $t\geq0$, we have
\begin{eqnarray}
t_1=\min_{l=1,\ldots,L} X_l
\end{eqnarray}
for non-preemptive SEDPT-NR, where $X_l$ is the chunk downloading time of thread $l$ if it does not switch to serve another chunk before completing the current chunk. Under other work-conserving policies, some thread may switch to serve another chunk.
If the thread has spent $\tau$ seconds on one chunk, the tail probability for completing the current chunk under service is $\mathbb{P}(X>t+\tau|X>\tau)$. On the other hand, the tail probability for switching to serve a new chunk is $\mathbb{P}(X>t)$. Since the chunk downloading time is \emph{i.i.d.} NLU, it is stochastically better to keep downloading the same chunk than switching to serve a new chunk. Therefore, $t_1$ under non-preemptive SEDPT-NR is stochastically smaller than that under any other work-conserving policy.

Next, suppose that $(t_1,t_2,\ldots, t_j)$ under non-preemptive SEDPT-NR are stochastically smaller than those under any other work-conserving policy. Let $R_l$ denote the remaining time for thread $l$ to download the current chunk after $t_j$. Under non-preemptive SEDPT-NR, since all $L$ threads are active at all time $t\geq0$, $t_{j+1}$ is determined as
\begin{eqnarray}\label{eq4}
t_{j+1} = \min_{l=1,\ldots,L}\left[t_j+R_l\right].
\end{eqnarray}
Under other work-conserving policies, some thread may switch to serve a new chunk before completing the current chunk. Similar as above, one can show that $(t_1,t_2,\ldots,t_{j+1})$ under non-preemptive SEDPT-NR are stochastically smaller than those under any other work-conserving policy. By induction, the chunk departure instants $(t_1,t_2,\ldots)$ under non-preemptive SEDPT-NR are stochastically smaller than those under any other work-conserving policy.

Finally, since the downloading times of different chunks are \emph{i.i.d.}, service idling only postpones chunk departure time. Hence, the chunk departure time instants will be larger under non-work-conserving policies. Therefore, $(t_1,t_2,\ldots)$ under non-preemptive SEDPT-NR are stochastically smaller than those under any other online policy.

\section{Proof of Lemma~\ref{thm6}} \label{app4}
We first construct a delay lower bound of $\overline{D}_{\text{opt}}$.
Consider a fixed sample path of the chunk departure instants $(t_1,t_2,\ldots)$. The request departure instants $(c_{1,\pi},c_{2,\pi},\ldots)$ are determined by the correspondence between the requests and the departed chunks. Define $r_i(t)$ as the number of remaining chunks to be downloaded after time $t$ for completing request $i$.
If each departed chunk belongs to an unfinished request $i$ with the smallest $r_i(t)$, the number of unfinished requests is minimized. By this, we obtain a lower bound on the sample-path average delay $\frac{1}{N}\sum_{i=1}^N (c_{i,\pi}-a_i)$. According to Lemma \ref{lem2}, the chunk departure instants $(t_1,t_2,\ldots)$ under non-preemptive SEDPT-NR are stochastically smaller than those under any other policy. By integrating $\frac{1}{N}\sum_{i=1}^N (c_{i,\pi}-a_i)$ over the distribution of $(t_1, t_2, \ldots)$ under non-preemptive SEDPT-NR, a delay lower bound of $\overline{D}_{\text{opt}}$ is obtained. On the other hand, non-preemptive SEDPT-NR provides an upper bound of $\overline{D}_{\text{opt}}$. The remaining task is to evaluate the delay gap between the delay lower bound and non-preemptive SEDPT-NR.

Next, we utilize the proof techniques of Theorem \ref{thm2} to evaluate the delay gap between non-preemptive SEDPT-NR and the above lower bound. For notational simplicity, we use policy $P$ to denote the above constructed policy that achieves a lower bound of $\overline{D}_{\text{opt}}$, and policy $NP$ to denote non-preemptive SEDPT-NR. We will show that \emph{for any time $t$ and any given sample path of chunk departures $(t_1,t_2,\ldots)$, policy $NP$ needs to download $L$ or fewer additional chunks after time $t$, so as to accomplish the same number of requests that are completed under policy $P$ during $(0,t]$}.

\begin{definition}\cite{Smith78}
The system state of policy $P$ is specified by an infinite vector $\vec{\beta}=(\beta_1,\beta_2,\ldots)$ with non-negative, non-increasing components. At any time, the coordinates of $\vec{\beta}$ are interpreted as follows: $\beta_1$ is the maximum number of remaining chunks among all requests, $\beta_2$ is the next greatest number of remaining chunks among all requests, and so on, with duplications being explicitly repeated. Suppose that there are $l$ unfinished requests in the system, then
\begin{eqnarray}
\beta_1\geq\beta_2\geq\ldots\geq\beta_l>0 = \beta_{l+1} =\beta_{l+2}=\ldots.
\end{eqnarray}
\end{definition}

\begin{definition}
The system state of \emph{non-preemptive SEDPT-NR (policy $NP$)} is specified by a pair of vectors $\{\vec{\alpha},\vec{\delta}\}$, where $\vec{\alpha}=(\alpha_1,\alpha_2,\ldots)$ and $\vec{\delta}=(\delta_1,\delta_2,\ldots)$ are two infinite vectors with non-negative components. At any time, the coordinates of $\vec{\alpha}$ and $\vec{\delta}$ are interpreted as follows: $\alpha_i$ is the number of chunks to be downloaded for completing the request associated to the $i$th coordinate, and $\delta_i$ is the number of threads assigned to serve the request associated to the $i$th coordinate such that $\sum_{i=1}^\infty \delta_i\leq L$. Suppose that there are $l$ unfinished requests in the system, then there exists an integer $m$ ($0\leq m\leq l$) such that the coordinates of $\vec{\alpha}$ and $\vec{\delta}$ satisfy
\begin{eqnarray}
&&\!\!\!\!\!\!\!\!\!\!\!\!\!\!\!\!\alpha_1-\delta_1\geq\ldots\geq \alpha_{m}-\delta_{m}>0= \alpha_{m+1}-\delta_{m+1}=\ldots\\
&&\!\!\!\!\!\!\!\!\!\!\!\!\!\!\!\!\alpha_i \left\{\begin{array}{l l} >0, &\text{if}~ i\leq l;\nonumber\\
  =0, &\text{if}~ i\geq l+1,\end{array}\right.\nonumber\\
  &&\!\!\!\!\!\!\!\!\!\!\!\!\!\!\!\!\delta_i \left\{\begin{array}{l l} \geq0, &\text{if}~ i\leq l;\\
  =0, &\text{if}~ i\geq l+1.\end{array}\right.
\end{eqnarray}
\end{definition}

\begin{lemma}\label{lem_SEDPT-NR0}
Let $\{\vec{\alpha}(t),\vec{\delta}(t),t\geq0\}$ be the state process of policy $NP$ and $\{\vec{\beta}(t),t\geq0\}$ be the state process of policy $P$. If $\vec{\alpha}(0)=\vec{\delta}(0) = \vec{\beta}(0) =0$, then for any given sample path of chunk departures $(t_1,t_2,\ldots)$, we have
\begin{eqnarray}
\sum_{i=j}^\infty [\alpha_i(t)- \delta_i(t)]\leq \sum_{i=j}^\infty \beta_i(t)
\end{eqnarray}
for all $t\geq0$ and j = $1,2,\ldots$
\end{lemma}
Lemma \ref{lem_SEDPT-NR0} can be obtained from the following lemmas:

\begin{lemma}\label{lem_SEDPT-NR1}
Suppose that, under policy $NP$, $\{\vec{\alpha}',\vec{\delta}'\}$ is obtained by completing a chunk at one of the $L$ threads in the system whose state is $\{\vec{\alpha},\vec{\delta}\}$. Further, suppose that, under policy $P$, $\vec{\beta}'$ is obtained by completing a chunk at one of the $L$ threads in the system whose state is $\vec{\beta}$.
If
\begin{eqnarray}\label{eq_SEDPT-NR_41}
\sum_{i=j}^\infty [\alpha_i - \delta_i]\leq \sum_{i=j}^\infty \beta_i, ~\forall~j=1,2,\ldots,
\end{eqnarray}
then
\begin{eqnarray}\label{eq_SEDPT-NR_40}
\sum_{i=j}^\infty [\alpha'_i - \delta'_i]\leq \sum_{i=j}^\infty \beta'_i, ~\forall~j=1,2,\ldots
\end{eqnarray}
\end{lemma}

\begin{proof}
If $\sum_{i=j}^\infty [\alpha'_i - \delta'_i]=0$, then the inequality \eqref{eq_non_prmp_40} follows naturally.

If $\sum_{i=j}^\infty [\alpha'_i - \delta'_i]>0$, suppose that, under policy $NP$, there are $m$ requests satisfying $\alpha_i - \delta_i>0$ at state $\{\vec{\alpha},\vec{\delta}\}$. After the chunk departure, the thread that just became idle will be assigned to serve a request associated to the smallest positive $\alpha_i - \delta_i$. This tells us that (i) $\alpha'_i - \delta'_i = \alpha_i - \delta_i$ for $i=1,2,\ldots,m-1$; (ii) $\alpha'_m - \delta'_m = \alpha_m - \delta_m-1$; and (iii) $\alpha'_i - \delta'_i = \alpha_i - \delta_i=0$ for $i=m+1,m+2,\ldots$
Since $\sum_{i=j}^\infty [\alpha'_i - \delta'_i]>0$, we have $j\leq m$.
Hence, $\sum_{i=j}^\infty [\alpha'_i - \delta'_i] = \sum_{i=j}^\infty [\alpha_i - \delta_i] -1 \leq \sum_{i=j}^\infty \beta_i -1\leq \sum_{i=j}^\infty \beta'_i$.
\end{proof}

\begin{lemma}\label{lem_SEDPT-NR2}
Suppose that, under policy $NP$, $\{\vec{\alpha}',\vec{\delta}'\}$ is obtained by adding a request with $b$ remaining chunks to the system whose state is $\{\vec{\alpha},\vec{\delta}\}$. Further, suppose that, under policy $P$, $\vec{\beta}'$ is obtained by adding a request with $b$ remaining chunks to the system whose state is $\vec{\beta}$.
If
\begin{eqnarray}
\sum_{i=j}^\infty [\alpha_i - \delta_i]\leq \sum_{i=j}^\infty \beta_i, ~\forall~j=1,2,\ldots,
\end{eqnarray}
then
\begin{eqnarray}
\sum_{i=j}^\infty [\alpha'_i - \delta'_i]\leq \sum_{i=j}^\infty \beta'_i, ~\forall~j=1,2,\ldots
\end{eqnarray}
\end{lemma}
The proof of Lemma \ref{lem_SEDPT-NR2} is the same with that of Lemma \ref{lem_non_prmp2}. Now, we are ready to prove Lemma~\ref{thm6}.

\begin{proof}[Proof of Lemma~\ref{thm6}]
As explained above, we only need to evaluate the delay gap between policy $NP$ and policy $P$. Let the evolution of the system state under some queueing discipline be on a space $(\Omega,\mathcal{F},P)$. We assume that the request arrival process $\{a_i,k_i,n_i\}_{i=1}^N$ is fixed for all $\omega\in\Omega$. Let $\{\vec{\alpha}(t), \vec{\delta}(t),t\geq0\}$ be the state process of policy $NP$ and $\{\vec{\beta}(t),t\geq0\}$ be the state process of policy $P$. Then, we have $\vec{\alpha}(0)=\vec{\beta}(0)=\vec{\delta}(0)=0$.

Suppose that under policy $NP$, there are $y$ request arrivals and $z$ request departures during $(0, t]$. Then, there are $y-z$ requests in the system at time $t$ such that $\sum_{i=y-z+1}^\infty \beta_i(t) = 0$. According to Lemma \ref{lem_SEDPT-NR0}, we have $\sum_{i=y-z+1}^\infty \alpha_i(t) \leq \sum_{i=y-z+1}^\infty \delta_i(t) \leq L$. Hence, under policy $NP$, the system still needs to download $L$ or fewer chunks associated to $\alpha_{y-z+1}(t), \alpha_{y-z+2}(t),\ldots$ after time $t$, in order to complete $z$ requests as in policy $P$. Further, $\sum_{i=y-z+1}^\infty \delta_i(t) \leq L$ tells us that the services of these chunks have already started by time $t$. Therefore, the average remaining downloading time of these chunks after time $t$ is no more than
\begin{eqnarray}
&&\overline{D}_{\text{extra}}\leq \mathbb{E}\left\{ \max_{l=1,\ldots, L} X_l\right\}.
\end{eqnarray}
Therefore, the delay gap between policy $NP$ and policy $P$ is no more than $\mathbb{E}\left\{ \max_{l=1,\ldots, L} X_l\right\}$, and Lemma \ref{thm6} is proven.
\end{proof}
\section{Proof of Theorem \ref{thm7}}\label{app5}
We will prove this theorem in three steps: in \emph{Step 1}, we will construct a virtual policy which provides delay lower bound of $\overline{D}_{\text{opt}}$; in \emph{Step 2}, we will compare the chunk departure sample paths of the constructed virtual policy and non-preemptive SEDPT-NR; in \emph{Step 3}, we will evaluate the delay gap between the delay lower bound and the average delay of non-preemptive SEDPT-NR. The details are provided in the sequel.

\emph{Step 1:}
We first construct a virtual policy which provides delay lower bound of $\overline{D}_{\text{opt}}$. Define $r(t)$ as the total number of remaining chunks to be downloaded for completing all the unfinished requests at time $t$.
We construct a virtual policy $P$ as follows: If $r(t)\geq L$ at time $t$, each thread is assigned to serve one chunk and will not switch to serve another chunk until it has completed the current chunk. If $0<r(t)< L$, suppose that there are $L-r(t)$ ``virtual'' chunks, such that each thread is assigned to serve one chunk and will not switch to serve another chunk until it has completed the current chunk. If $r(t)= 0$, all $L$ threads are idle. Further, under the virtual policy $P$, each departed chunk belongs to an unfinished request with the fewest remaining chunks. Similar to Lemma \ref{lem2}, we can obtain the following result:

\begin{lemma}\label{lemG1} If the chunk downloading time is {i.i.d.} NLU, then for given request parameters $N$ and $(a_i,k_i,n_i)_{i=1}^N$, the constructed chunk departure instants $(t_1,t_2,\ldots)$ of policy $P$ are stochastically smaller than those under any online policy.
\end{lemma}
\begin{proof}
We first compare the chunk departure times among the class of work-conserving policies.

Let us consider the departure time of the first chunk $t_1$. Because $a_1=s_1=0$ and all $L$ threads are active for $t\geq0$, we have
\begin{eqnarray}
t_1=\min_{l=1,\ldots,L} X_l
\end{eqnarray}
for the constructed chunk departures, where $X_l$ is the chunk downloading time of thread $l$ if it does not switch to serve another chunk before completing the current chunk. Under other work-conserving policies, some thread may switch to serve another chunk. We have shown that, if the chunk downloading time is \emph{i.i.d.} NLU, it is stochastically better to keep downloading the same chunk than switching to serve a new chunk. Therefore, $t_1$ under policy $P$ is stochastically smaller than that under any work-conserving policy.

Next, suppose that the constructed chunk departure instants $(t_1,t_2,\ldots, t_j)$ of policy $P$ are stochastically smaller than those under any work-conserving policy. Let $R_l$ denote the remaining downloading time of thread $l$ for serving the current chunk after time $\max\{s_{j+1},t_j\}$. Under policy $P$, all $L$ threads are active after time $\max\{s_{j+1},t_j\}$. Hence, $t_{j+1}$ is determined as
\begin{eqnarray}\label{eq4}
t_{j+1} = \min_{l=1,\ldots,L}\left[\max\{s_{j+1},t_j\}+R_l\right].
\end{eqnarray}
Under other work-conserving policies, some thread may switch to serve a new chunk before completing the current chunk. Similar with the above discussions, one can show that the chunk departure instants $(t_1,t_2,\ldots,t_{j+1})$  of policy $P$ are stochastically smaller than those under any work-conserving policy. By induction, the constructed chunk departure instants $(t_1,t_2,\ldots)$ of policy $P$ are stochastically smaller than those under any work-conserving policy.

Finally, since the downloading times of different chunks are \emph{i.i.d.}, service idling only postpones chunk departure time. Hence, the chunk departure times will be larger under non-work-conserving policies. Therefore, the constructed chunk departure instants $(t_1,t_2,\ldots)$  of policy $P$ are stochastically smaller than those under any online policy.
\end{proof}

Under policy $P$, each departed chunk belongs to an unfinished request with the fewest remaining chunks, such that the number of unfinished requests is minimized. According to Lemma \ref{lemG1}, the constructed chunk departure instants $(t_1,t_2,\ldots)$  of policy $P$ are stochastically smaller than those under any online policy. By taking the expectation over the distribution of $(t_1, t_2, \ldots)$, one can show that the virtual policy $P$ provides a delay lower bound of $\overline{D}_{\text{opt}}$. On the other hand, non-preemptive SEDPT-NR provides an upper bound of $\overline{D}_{\text{opt}}$. The remaining task is to evaluate the delay gap between policy $P$ and non-preemptive SEDPT-NR.

\emph{Step 2:} We now study the chunk departure sample paths of policy $P$ and non-preemptive SEDPT-NR. For notational simplicity, we use policy $NP$ to denote non-preemptive SEDPT-NR. Similar to the proof of Lemma \ref{thm6}, we define the system states of policy $P$ and policy $NP$. Let $\{\vec{\alpha}(t),\vec{\delta}(t),t\geq0\}$ be the state process of policy $NP$ and $\{\vec{\beta}(t),t\geq0\}$ be the state process of policy $P$. Suppose that $\vec{\alpha}(0) = \vec{\delta}(0)=\vec{\beta}(0)=0$.

\begin{lemma}\label{lem_Yin}
If $\vec{\alpha}(0) = \vec{\delta}(0)=\vec{\beta}(0)=0$, then for any chunk departure sample path of policy $NP$, there exists a chunk departure sample path of policy $P$, such that for any time $t$ the number of chunks downloaded during $(0,t]$ under the sample path of policy $P$ is no more than $L-1$ plus the number of chunks downloaded during $(0,t]$ under the sample path of policy $NP$, i.e.,
\begin{eqnarray}\label{eq_thm7}
\sum_{i=1}^\infty\alpha_i(t) \leq \sum_{i=1}^\infty\beta_i(t) + L-1,~\forall~t\geq0.
\end{eqnarray}
\end{lemma}
\begin{proof}
{We partition the system service duration of policy $NP$ into a sequence of time intervals $(\tau_1, \nu_1]$, $(\nu_1, \tau_2]$, $(\tau_2, \nu_2]$, $(\nu_2, \tau_3]$, $\ldots$, such that $r(t)\leq L-1$ for $t\in (\tau_i, \nu_i]$ and $r(t)\geq L$ for $t\in (\nu_i, \tau_{i+1}]$ for $i=1,2,\ldots$ Therefore, under policy $NP$, at most $L-1$ threads are active during the intervals $(\tau_i, \nu_i]$ and all $L$  threads are active during the intervals $(\nu_i, \tau_{i+1}]$. We construct a ``virtual'' policy $Q$ based on policy $NP$: After time $\tau_i$, there are at most $L-1$ remaining chunks to be downloaded. Under policy $Q$, these remaining chunks are completed immediately after time $\tau_i$ such that the $L$ threads are idle during $(\tau_i, \nu_i]$. During $(\nu_i, \tau_{i+1}]$, policy $Q$ is defined according to the same principle of policy $P$: ``virtual chunks'' are used when there are less than $L$ remaining chunks such that all $L$ threads are active under policy $Q$ until there is no remaining chunk to download.
The system state of policy $Q$ is specified by an infinite vector $\vec{\gamma}=(\gamma_1,\gamma_2,\ldots)$ with non-negative, non-increasing components. At any time, the coordinates of $\vec{\gamma}$ are interpreted as follows: $\gamma_1$ is the maximum number of remaining chunks among all requests, $\gamma_2$ is the next greatest number of remaining chunks among all requests, and so on, with duplications being explicitly repeated.

Next, we prove that
\begin{eqnarray}\label{eq_Yin}
\sum_{i=1}^\infty\alpha_i(t) \leq \sum_{i=1}^\infty\gamma_i(t) + L-1
\end{eqnarray}
for all $t\geq0$.
During $(\tau_i, \nu_i]$, we have $\sum_{i=1}^\infty\alpha_i(t)=r(t)\leq L-1$ and $\sum_{i=1}^\infty\gamma_i(t)=0$. Hence, \eqref{eq_Yin} holds during $(\tau_i, \nu_i]$. At time $\nu_i$, policy $NP$ has at most $L-1$ extra chunks, compared to policy $Q$. Further, the  $L$ threads of policy $NP$ start downloading earlier than time $\nu_i$, while the $L$ threads of policy $Q$ start downloading exactly at time $\nu_i$. Hence, Therefore, \eqref{eq_Yin} must hold during $(\nu_i, \tau_{i+1}]$. By induction, \eqref{eq_Yin} holds for all $t\geq0$.

Further, we show that there exists a chunk departure sample path of policy $P$ such that
\begin{eqnarray}\label{eq_Yin1}
\sum_{i=1}^\infty\gamma_i(t) \leq \sum_{i=1}^\infty\beta_i(t), ~\forall~t\geq0.
\end{eqnarray}
During $(\tau_i, \nu_i]$, policy $Q$ satisfies $\sum_{i=1}^\infty\gamma_i(t) =0$ and \eqref{eq_Yin1} follows. During $(\nu_i, \tau_{i+1}]$, policy $Q$ satisfies the same principle as policy $P$, except for their different initial states at time $\nu_i$. In particular, policy $Q$ has no chunk to download before time $\nu_i$ and policy $P$ may have some chunks not completed yet before time $\nu_i$. Therefore, policy $P$ needs to complete these remaining chunks to have the same state with policy $Q$. Since policy $P$ and policy $Q$ satisfy the same principle, there must exist a chunk departure sample path of policy $P$ such that \eqref{eq_Yin1} holds during $(\nu_i, \tau_{i+1}]$. By induction, \eqref{eq_Yin1} holds for all $t\geq0$. Combining \eqref{eq_Yin} and \eqref{eq_Yin1}, Lemma \ref{lem_Yin} follows.}
\end{proof}

\emph{Step 3:} We will show that \emph{for any time $t$ and the chunk departure sample paths constructed above, policy $NP$ needs to download $2L-1$ or fewer additional chunks after time $t$, so as to accomplish the same number of requests that are completed under policy $P$ during $(0,t]$}. Towards this goal, we need to prove the following lemma:

\begin{lemma}\label{lemG4}
Let $\{\vec{\alpha}(t),\vec{\delta}(t),t\geq0\}$ be the state process of policy $NP$ and $\{\vec{\beta}(t),t\geq0\}$ be the state process of policy $P$. If $\vec{\alpha}(0) = \vec{\delta}(0)=\vec{\beta}(0)=0$, then under the chunk departure sample paths of policy $NP$ and policy $P$ mentioned above, we have
\begin{eqnarray}
\sum_{i=1}^\infty \beta_i(t)+ \sum_{i=j}^\infty [\alpha_i(t)- \delta_i(t)]\leq \sum_{i=j}^\infty \beta_i(t) + \sum_{i=1}^\infty \alpha_i(t)
\end{eqnarray}
for all $t\geq0$ and j = $1,2,\ldots$
\end{lemma}
Lemma \ref{lemG4} can be easily obtained from the following two lemmas:

\begin{lemma}\label{lem_G1}
Suppose that, under policy $NP$, the system state at time $t$ is $\{\vec{\alpha},\vec{\delta}\}$ and at time $t+\Delta t$ is $\{\vec{\alpha}',\vec{\delta}'\}$. Further, suppose that, under policy $P$, the system state at time $t$ is $\vec{\beta}$ and at time $t+\Delta t$ is $\vec{\beta}'$. If (i) no arrivals occur during the interval $(t, t + \Delta t]$ and (ii)
\begin{eqnarray}\label{eq_G_41}
\sum_{i=1}^\infty \beta_i + \sum_{i=j}^\infty [\alpha_i - \delta_i]\leq \sum_{i=j}^\infty \beta_i+ \sum_{i=1}^\infty \alpha_i, ~\forall~j=1,2,\ldots,\!\!
\end{eqnarray}
then
\begin{eqnarray}\label{eq_G_40}
\sum_{i=1}^\infty \beta'_i+\sum_{i=j}^\infty [\alpha'_i - \delta'_i]\leq \sum_{i=j}^\infty \beta'_i+ \sum_{i=1}^\infty \alpha'_i, ~\forall~j=1,2,\ldots
\end{eqnarray}
\end{lemma}

\begin{proof}
If $\sum_{i=j}^\infty [\alpha'_i - \delta'_i]= 0$, then the inequality \eqref{eq_non_prmp_40} follows naturally.

If $\sum_{i=j}^\infty [\alpha'_i - \delta'_i]> 0$, suppose that $b$ chunks are downloaded under policy $NP$ during $(t, t+\Delta t]$, and $d$ chunks are downloaded under policy $P$. Then, we have
\begin{eqnarray}\label{eq_G_412}
\sum_{i=1}^\infty\alpha_i - \sum_{i=1}^\infty\alpha_i' = b,\\
\sum_{i=1}^\infty\beta_i - \sum_{i=1}^\infty\beta_i' = d.
\end{eqnarray}
Further, under policy $NP$, the smallest and yet positive $\alpha_i - \delta_i$ will decrease by one after each chunk departure. Hence, we have
\begin{eqnarray}\label{eq_G_411}
\sum_{i=j}^\infty [\alpha'_i - \delta'_i] = \sum_{i=j}^\infty [\alpha_i - \delta_i] -b.
\end{eqnarray}
Using \eqref{eq_G_41}, \eqref{eq_G_412}-\eqref{eq_G_411}, we obtain $\sum_{i=1}^\infty \beta'_i + \sum_{i=j}^\infty [\alpha'_i - \delta'_i] = \sum_{i=1}^\infty \beta'_i + \sum_{i=j}^\infty [\alpha_i - \delta_i] -b = \sum_{i=1}^\infty \beta'_i + \sum_{i=j}^\infty [\alpha_i - \delta_i] + \sum_{i=1}^\infty\alpha_i' - \sum_{i=1}^\infty\alpha_i \leq \sum_{i=1}^\infty\beta_i' + \sum_{i=1}^\infty \alpha'_i + \sum_{i=j}^\infty \beta_i - \sum_{i=1}^\infty \beta_i = \sum_{i=1}^\infty \alpha'_i + \sum_{i=j}^\infty \beta_i - d \leq \sum_{i=1}^\infty \alpha'_i + \sum_{i=j}^\infty \beta'_i$.
\end{proof}

\begin{lemma}\label{lem_G2}
Suppose that, under policy $NP$, $\{\vec{\alpha}',\vec{\delta}'\}$ is obtained by adding a request with $b$ remaining chunks to the system whose state is $\{\vec{\alpha},\vec{\delta}\}$. Further, suppose that, under policy $P$, $\vec{\beta}'$ is obtained by adding a request with $b$ remaining chunks to the system whose state is $\vec{\beta}$.
If
\begin{eqnarray}
\sum_{i=1}^\infty \beta_i +\sum_{i=j}^\infty [\alpha_i - \delta_i]\leq \sum_{i=j}^\infty \beta_i+  \sum_{i=1}^\infty \alpha_i, ~\forall~j=1,2,\ldots,
\end{eqnarray}
then
\begin{eqnarray}
\sum_{i=1}^\infty \beta_i'+\sum_{i=j}^\infty [\alpha'_i - \delta'_i]\leq \sum_{i=j}^\infty \beta'_i+ \sum_{i=1}^\infty \alpha'_i, ~\forall~j=1,2,\ldots
\end{eqnarray}
\end{lemma}
The proof of Lemma \ref{lem_G2} is quite similar with that of Lemma \ref{lem_non_prmp2} and is thus omitted. We now prove Theorem \ref{thm7}.
\begin{proof}[Proof of Theorem \ref{thm7}]
As explained above, we only need to evaluate the delay gap between policy $NP$ and policy $P$. Let the evolution of the system state under some queueing discipline be on a space $(\Omega,\mathcal{F},P)$. We assume that the request arrival process $\{a_i,k_i,n_i\}_{i=1}^N$ is fixed for all $\omega\in\Omega$. Let $\{\vec{\alpha}(t), \vec{\delta}(t),t\geq0\}$ be the state process of policy $NP$ and $\{\vec{\beta}(t),t\geq0\}$ be the state process of policy $P$. Then, we have $\vec{\alpha}(0)=\vec{\beta}(0)=\vec{\delta}(0)=0$.

Suppose that under policy $NP$, there are $y$ request arrivals and $z$ request departures during $(0, t]$. Then, there are $y-z$ requests in the system at time $t$ such that $\sum_{i=y-z+1}^\infty \beta_i(t) = 0$. According to Lemma \ref{lemG4} and \eqref{eq_thm7}, we have $\sum_{i=y-z+1}^\infty \alpha_i(t) \leq \sum_{i=y-z+1}^\infty \delta_i(t) + \sum_{i=1}^\infty [\alpha_i(t) - \beta_i (t)] \leq 2L - 1$.
Hence, under policy $NP$, the system still needs to download $2L-1$ chunks  after time $t$, in order to complete $z$ requests as in policy $P$. Therefore, the average downloading time of these extra chunks after time $t$ is no more than
\begin{eqnarray}
&&\overline{D}_{\text{extra}}\leq \mathbb{E}\left\{ \max_{l=1,\ldots, L} X_l\right\} + \mathbb{E}\left\{ \max_{l=1,\ldots, L-1} X_l\right\}.
\end{eqnarray}
Hence, the delay gap between policy $NP$ and policy $P$ is no more than $\mathbb{E}\left\{ \max_{l=1,\ldots, L} X_l\right\}+ \mathbb{E}\left\{ \max_{l=1,\ldots, L-1} X_l\right\}$. By this, Theorem \ref{thm7} is proven.
\end{proof}

\section{Proof of Theorem \ref{thm8}} \label{app5.5}
When preemption is allowed, the proof of Theorem \ref{thm7} can be directly used to show that \eqref{eq_7} still holds, with $\overline{D}_{\text{opt}}$ representing the optimal delay performance in the preemptive case. Further, preemptive SEDPT-WCR can achieve a shorter average delay than non-preemptive SEDPT-NR when preemption is allowed. Then, Theorem \ref{thm8} follows.

\section{Proof of Lemma \ref{lem3}} \label{app6}
We first compare the chunk departure time sequence among the class of work-conserving policies.
 Since $d_{\min}\geq L$, all $L$ threads are kept active whenever there are unfinished requests.

Let us consider the departure time of the first chunk $t_1$. Since $a_1=s_1=0$, for any non-preemptive work-conserving policy, we have
\begin{eqnarray}
t_1=\min_{l=1,\ldots,L} X_l.
\end{eqnarray}
Therefore, the distribution of $t_1$ is invariant under any non-preemptive work-conserving policy.

Next, suppose that $(t_1,t_2,\ldots, t_j)$ under non-preemptive SEDPT-R are stochastically smaller than those under any other work-conserving policy. Let $\tau_l$ denote the time that thread $l$ has spent on the current chunk up to time $t_j$, and $R_l$ denote the remaining time for thread $l$ to download the current chunk after time $t_j$. The tail distribution of $R_l$ is given by
\begin{eqnarray} \label{eq3}
\mathbb{P}(R_l>\gamma|\tau_l = \tau) = \mathbb{P}( X > \gamma +\tau | X>\tau).
\end{eqnarray}
By \eqref{eq3} and the condition that the chunk downloading time distribution is NSU, the remaining downloading time $R_l$  of the case $\tau_l=0$ is stochastically smaller  than that  of the case $\tau_l=\tau>0$. In other words, the remaining downloading time $R_l$ is stochastically smaller if thread $l$ switches to download a new chunk at time $t_j$.
For any non-preemptive work-conserving policy, $t_{j+1}$ is determined as
\begin{eqnarray}\label{eq4}
t_{j+1} = \min_{l=1,\ldots,L}\left[\max\{s_{j+1}, t_j\}+R_l\right].
\end{eqnarray}
Hence, $(t_1,t_2,\ldots,t_{j+1})$ is stochastically smaller if all $L$ threads switch to download a new chunk at time $t_j$. This only occurs under SEDPT-R, where all $L$ threads are assigned to serve the same request. Therefore, $(t_1,t_2,\ldots, t_{j+1})$ under non-preemptive SEDPT-R are stochastically smaller than those under any other work-conserving policy.

By induction, $(t_1,t_2,\ldots, t_{N})$ under non-preemptive SEDPT-R are stochastically smaller than those under any other work-conserving policy.

Finally, since the downloading times of different chunks are \emph{i.i.d.}, service idling only postpones chunk departure time. Hence, the chunk departure times will be larger under non-work-conserving policies. Therefore, $(t_1,t_2,\ldots, t_{N})$ under non-preemptive SEDPT-R are stochastically smaller than those under any other online policy.

\section{Proof of Theorem~\ref{thm5}}\label{app7}
Since $k_i=1$ for all $i$, each file only has one remaining chunk. Hence, the file departure process $(c_{1,\pi},c_{2,\pi},\ldots,c_{N,\pi})$ is a permutation of $(t_1,t_2,\ldots,t_N)$ and
\begin{eqnarray}
\sum_{i=1}^N \mathbb{E}\left\{t_i\right\} = \sum_{i=1}^N \mathbb{E}\left\{c_{i,\pi}\right\}.
\end{eqnarray}
In Lemma \ref{lem3}, it was shown that the chunk departure instants $(t_1,t_2,\ldots, t_N)$ under non-preemptive SEDPT-R are stochastically smaller than those under any other online policy. Therefore, non-preemptive SEDPT-R minimizes $\sum_{i=1}^N \mathbb{E}\left\{t_i\right\}$ \cite{StochasticOrderBook}. By this, Theorem \ref{thm5} is proven.
\fi

\end{document}